%% file: jtsp16chepuri_arxiv.tex
\documentclass[10pt,twocolumn,twoside] {IEEEtran}
\usepackage{cite,url}
\usepackage{amsfonts}
\usepackage{fancyhdr}
\usepackage{listings}%
\usepackage{arydshln}
\usepackage{amsmath,amssymb,amsthm,mathdots}%
\usepackage{graphicx,psfrag, caption,subcaption}%
\usepackage{tikz}%
\usepackage{listings}%
\usepackage{arydshln}
\usepackage{algorithm,algorithmic}%
\newtheorem{mylem}{Lemma}% \newtheorem{name}{label}[section]
\newtheorem{mydef}{Definition}%
\newtheorem*{myprob}{Problem}%
\newtheorem{mytheo}{Theorem}%
\newtheorem{mycor}{Corollary}%
\newtheorem{myrem}{Remark}%

\DeclareMathOperator*{\argmin}{arg\,min}
\DeclareMathOperator*{\argmax}{arg\,max}
\DeclareMathOperator*{\diag}{\mathrm{diag}}

\def\x{{\boldsymbol x}}
\def\w{{\boldsymbol w}}

\def\Y{{\boldsymbol Y}}
\def\a{{\boldsymbol a}}
\def\p{{\boldsymbol p}}
\def\h{{\boldsymbol h}}
\def\n{{\boldsymbol n}}
\def\thetab{{\boldsymbol \theta}}
\def\vec{{\rm vec}}
\def\om{{\omega}}
\def\Sb{{\boldsymbol S}}
\def\H{{\boldsymbol H}}
\def\R{{\boldsymbol R}}
\def\Q{{\boldsymbol Q}}

\def\b{{\boldsymbol b}}
\def\r{{\boldsymbol r}}
\def\y{{\boldsymbol y}}
\def\G{{\boldsymbol G}}
\def\A{{\boldsymbol A}}
\def\B{{\boldsymbol B}}
\def\Phib{{\boldsymbol \Phi}}
\def\Psib{{\boldsymbol \Psi}}
\def\U{{\boldsymbol U}}
\def\u{{\boldsymbol u}}

\title{Graph Sampling for Covariance Estimation}
\author{Sundeep Prabhakar Chepuri,~\IEEEmembership{Member,~IEEE,} 
       Geert~Leus,~\IEEEmembership{Fellow,~IEEE} 
\thanks{The authors are with the Faculty of Electrical Engineering, Mathematics and Computer Science, Delft University of Technology, The Netherlands. Email:~\{s.p.chepuri;g.j.t.leus\}@tudelft.nl.}
\thanks{This work was supported by the KAUST-MIT-TUD consortium grant~{OSR-2015-Sensors-2700}. 
A conference precursor of this manuscript appeared in the the Ninth IEEE Sensor Array and
Multichannel Signal Processing Workshop  (SAM), Rio de Janeiro, Brazil, June 2016~\cite{sundeep16sam}.}
}

\begin{document}
\maketitle
\begin{abstract}
%%%%
% 
%
%
%%%%
In this paper the focus is on subsampling as well as reconstructing the second-order statistics of signals residing on nodes of arbitrary undirected graphs.
Second-order stationary graph signals may be obtained by graph filtering zero-mean white noise and they admit a well-defined power spectrum whose shape is determined by the frequency response of the graph filter. Estimating the graph power spectrum forms an important component of stationary graph signal processing and related inference tasks such as Wiener prediction or inpainting on graphs. The central result of this paper is that by sampling a significantly smaller subset of vertices and using simple least squares, we can reconstruct the second-order statistics of the graph signal from the subsampled observations, and more importantly, without any spectral priors. To this end, both a nonparametric approach as well as parametric approaches including moving average and autoregressive models for the graph power spectrum are considered. The results specialize for undirected circulant graphs in that the graph nodes leading to the best compression rates are given by the so-called minimal sparse rulers. A near-optimal greedy algorithm is developed to design the subsampling scheme for the non-parametric and the moving average models, whereas a particular subsampling scheme that allows linear estimation for the autoregressive model is proposed. Numerical experiments on synthetic as well as real datasets related to climatology and processing handwritten digits are provided to demonstrate the developed theory. 

\end{abstract}
{\IEEEkeywords Graph signal processing, stationary graph signals, sparse sampling, graph power spectrum estimation, compressive covariance sensing.}

\section{Introduction}

Graphs are mathematical objects that can be used for describing and explaining relationships in complex datasets, which appear commonly in modern data analysis. The nodes of the graph denote the entities themselves and the edges encode the pairwise relationship between these
entities. Some examples of such complex-structured data beyond traditional time-series include gene regulatory networks~\cite{barabasi2004network}, brain networks~\cite{bullmore2009complex}, transportation networks~\cite{guimera2005worldwide}, social and economic networks~\cite{jackson2010social}, and so on. Processing signals residing on the nodes of a graph taking into account the relationships between them as explained by the edges of the graph is recently receiving a significant amount of interest. In particular, generalizing as well as drawing parallels of classical time-frequency analysis tools to {\it graph data analysis} while incorporating the irregular structure on which the graph signals are defined is an emerging area of research~\cite{shuman2013Emerging,sandryhaila2014big}. 

Graph signals could be stochastic in nature and they can be modeled as the output of a graph filter~\cite{sandryhaila2013discrete} whose input is also a random {signal} (e.g., white noise). We are interested in sampling and processing stationary graph signals, which are stochastic signals defined on graphs with
second-order statistics that are invariant similar to time series, but in the graph setting. Second-order stationary graph signals are characterized by a well-defined {\it graph power spectrum}. They can be generated by { graph} filtering white noise (or any other stationary graph signal) and the graph power spectrum of the filtered signal will be characterized by the squared magnitude of the frequency response of the filter{; see~\cite{benjamin15eusipco,perraudin2016stationary,marques2016stationary,girault2015signal}.}  

The second-order statistics of graph signals, or equivalently the graph power spectrum, are essential to solve inference problems on graphs in the Bayesian setting such as smoothing, prediction, inpainting, and deconvolution; {see~\cite{girault2014semi} and~\cite{perraudin2016stationary} for some Bayesian inference problems}. These inference problems are solved by designing optimum (in the minimum mean squared error sense) Wiener-like filters and the graph power spectrum forms a crucial component of such filter designs.  
In order to compute the graph power spectrum, traditional methods require { the processing of signals on all graph nodes.}
The sheer quantity of data and scale of the graph often inhibit this reconstruction method. Therefore, the main question that we address in this paper is,
{\it can we reconstruct the graph power spectrum by observing a small subset of graph nodes?} 

%an eigenvalue decomposition of the matrix that explains the topology of the graph. 
\subsection{Related works and main results}

The notion of stationarity of signals on graphs and related definitions can be found in~\cite{benjamin15eusipco,perraudin2016stationary,marques2016stationary,girault2015signal}, and it will be briefly explained in the next section as well.  Several techniques for graph power spectrum estimation have been discussed in~\cite{perraudin2016stationary} and~\cite{marques2016stationary}, and they are based on observations from all the nodes. In this paper, we consider the problem of reconstructing the second-order statistics of signals on graphs, but from subsampled observations. 
The fact that we are reconstructing the graph power spectrum, instead of the graph signal, enables us to subsample the graph signal (or sparsely sample the graph nodes), even without any spectral priors (e.g., sparsity, bandlimited with known support). This is a new and different perspective as compared to subsampling for graph signal reconstruction~\cite{anis2014towards,tsitsvero2015uncertainty,varma2015spectrum,chen2015discrete}, which imposes some spectral prior that enables graph signal reconstruction. The proposed concept basically generalizes the field of compressive covariance sensing~\cite{ariananda2012compressive,romero2013wideband,Romero16CCSspm} to the graph setting.

The aim of this paper is to reconstruct second-order statistics of stationary graph signals from observations available at a few nodes using simple reconstruction methods such as least squares. The contributions are summarized as the following main results:
\begin{itemize}
\item {\it Non-parametric approach}: Without any spectral priors, second-order statistics of length-$N$ stationary graph signals can be recovered using least squares from a reduced subset of $\mathcal{O}(\sqrt{N})$ observations, i.e., by observing $\mathcal{O}(\sqrt{N})$ graph nodes. In this case, the processing is done in the graph spectral domain. 
\item {\it Circulant graphs:}  As a special case, when the graphs are circulant, the identifiability results are elegant. That is, the subset of nodes resulting in the best compression rates are given by the so-called {\it minimal sparse rulers}. This is reminiscent of compressive covariance sensing~\cite{Romero16CCSspm} for data that reside on a regular support such as time series, which is a specific instance of a circulant graph. 
\item {\it Parametric approach}: It is also possible to model the graph power spectrum using a small number of parameters, e.g., the graph signals may be modeled by {\it moving average} or {\it autoregressive} graph filters. The reconstruction of the second-order statistics of the graph signal then boils down to the estimation of moving average or autoregressive coefficients. Such a parameterization allows for a higher compression. 
When the graph power spectrum is modeled using a moving average graph filter, the second-order statistics can be recovered using least squares from 
$\mathcal{O}(\sqrt{Q})$ observations, where $Q = \min\{2L-1, N\}$ with $L$ being the number of moving average filter coefficients. When the graph power spectrum is modeled using an autoregressive graph filter, $P$ autoregressive filter coefficients can be recovered using {\it linear least squares} by observing $\mathcal{O}(P)$ nodes. 

\item {\it Subsampler design:}  The proposed samplers are deterministic and they perform {\it node} subsampling. Subsampler design, therefore, becomes a discrete combinatorial optimization problem. For the spectral domain and moving average case, the subsampler can be designed using a near-optimal greedy algorithm. However, for the autoregressive approach, the sampler design depends also on (unobserved) data, and thus a mean squared error optimal design is not possible. This is due to the fact that we restrict ourselves to a low-complexity linear estimator for the autoregressive filter coefficients. Nevertheless, we present a suboptimal technique to design a subsampler for the autoregressive case as well.    
\end{itemize}

\subsection{Outline and notation}
The remainder of the paper is organised as follows. The preliminary concepts of graph signal processing are discussed in Section~\ref{sec:prelims}.  
The proposed least squares based reconstruction of the second-order statistics based on the subsampled observations are discussed in Section~\ref{sec:subsampling}. Connections of compressive covariance sensing for time-series with sensing data residing on circulant graphs are discussed in
Section~\ref{sec:cyclicgraphs}. In Section~\ref{sec:parametric}, the graph power spectrum is represented with a small number of parameters under moving average and autoregressive models, and these parameters are then reconstructed using least squares  from subsampled observations. In Section~\ref{sec:finitedata}, we discuss the validity of the results provided in this paper for finite data records. Under the assumption that the data follows a Gaussian distribution, the maximum likelihood estimator and  the related Cram\'er-Rao bound are also derived. In Section~\ref{sec:subsampdesgn}, the design of sparse sampling matrices based on low-complexity greedy algorithms is discussed. A few examples to illustrate the proposed framework are provided in Section~\ref{sec:numericalresults}.  Finally, the paper concludes with Section~\ref{sec:conc}.

The notation used in this paper is described as follows. Upper (lower) boldface letters are used for matrices (column vectors). {Overbar $\bar{(\cdot)}$ denotes complex conjugation, $(\cdot)^T$ denotes the transpose, and $(\cdot)^H$ denotes the complex conjugate (Hermitian) transpose. $(\cdot)^{-T}$ is a shorthand notation for $\left((\cdot)^{-1}\right)^T$.} $\mathrm{diag}[\cdot]$ refers to a diagonal matrix with its argument on the main diagonal. ${\rm diag_r}[\cdot]$ represents a diagonal matrix with the argument on its diagonal, but with the all-zero rows removed. ${\boldsymbol 1}$ $({\boldsymbol 0})$ denotes the vector of all ones (zeros). ${\boldsymbol I}$ is an identity matrix. $\mathbb{E}\{\cdot\}$ denotes the expectation operation.   The $\ell_0$-(quasi) norm of ${\boldsymbol w}  = [w_1,w_2,\ldots,w_N]^T$ refers to the number of non-zero entries in ${\boldsymbol w}$, i.e., ${\|{\boldsymbol w}\|}_0 := |\{n\,: \, w_n \neq 0\}|$. The $\ell_1$-norm of ${\boldsymbol w}$ is denoted by ${\|{\boldsymbol w}\|}_1 = \sum_{n=1}^N |w_n|$.  The notation $\thicksim$ is read as ``is distributed according to". Unless and otherwise noted, logarithms are natural.
%The indicator function is written as $\mathds{1}_{\{\cdot\}}$.
${\rm tr}\{\cdot\}$ is the matrix trace operator. ${\rm det}\{\cdot\}$ is the matrix determinant. ${\rm rank} (\cdot)$ denotes the rank of a matrix. $\lambda_{\rm min}\{{\boldsymbol A}\}$ ($\lambda_{\rm max}\{{\boldsymbol A}\}$) denotes the minimum (maximum) eigenvalue of a symmetric matrix ${\boldsymbol A}$. ${\boldsymbol A} \succeq {\boldsymbol B}$ means that ${\boldsymbol A} - {\boldsymbol B}$ is a positive semidefinite matrix. $\mathbb{S}^{N}$ ($\mathbb{S}^{N}_{+}$) denotes the set of symmetric (symmetric positive semi-definite) matrices of size $N \times N$. $|\mathcal{U}|$ denotes the cardinality of the set $\mathcal{U}$. $\otimes$ denotes the Kronecker product, $\circ$ denotes the Khatri-Rao or columnwise Kronecker product, and ${\rm vec}(\cdot)$ refers to the matrix vectorization operator. For a full column rank tall matrix ${\boldsymbol A}$, the left inverse is given by ${\boldsymbol A}^\dagger = ({\boldsymbol A}^H{\boldsymbol A})^{-1}{\boldsymbol A}^H$. 
The column span of $\A$ and row null space of $\A$ are denoted by ${\rm ran}(\A)$ and ${\rm null}(\A)$, respectively.
Properties that are frequently used in this paper:
\begin{itemize}
\item  ${\rm vec}({\boldsymbol A}{\boldsymbol B}{\boldsymbol C})=({\boldsymbol C}^T \otimes {\boldsymbol A}){\rm vec}({\boldsymbol B});$
\item ${\rm vec}({\boldsymbol A}{\rm diag}[{\boldsymbol b}]{\boldsymbol C})=({\boldsymbol C}^T \circ {\boldsymbol A}){\boldsymbol b}.$
\end{itemize}

\section{Preliminaries} \label{sec:prelims}

In this section, we introduce some preliminary concepts related to deterministic and stochastic signals defined on graphs.

\subsection{Graph signals and filtering}

Consider a dataset with $N$ elements denoted as $\x \in \mathbb{C}^N$, which live on an irregular structure represented by an {undirected} graph 
$\mathcal{G} = (\mathcal{V},\mathcal{E})$, where the vertex set $\mathcal{V} = \{v_1,\cdots,v_N\}$ denotes the set of nodes, and the edge set $\mathcal{E}$ reveals any connection between the nodes{, i.e., $(i,j)\in\mathcal{E}$ means that node $i$ is connected to node $j$}. The $n$th entry of $\x$, i.e., $x_n$, is indexed by node $v_n$ of the graph $\mathcal{G}$. Therefore, we refer to the dataset $\x$ as a  length-$N$ {\it graph signal}. 

%Let ${\boldsymbol A} \in \mathbb{C}^{N \times N}$ be the adjacency matrix of graph $\mathcal{G}$ with a nonzero $(i,j)$th entry $a_{i,j}$ denoting the strength of the edge connecting the $i$th node and the $j$th node, while the entry is set to zero if no edge exists between the $i$th node and the $j$th node. We assume undirected graphs, for which we have $a_{i,j} = a_{j,i}$. The degree of the $i$th node is defined as $d_i = \sum_{j=1}^M a_{i,j}$. A related operator, the so-called graph Laplacian, is defined as 
%$
%{\boldsymbol L} = {\boldsymbol D} - {\boldsymbol A} \, \in \, \mathbb{S}^N,
%$ 
%where ${\boldsymbol D} = {\rm diag}[d_1,d_2,\cdots,d_N] \in \mathbb{R}^{N\times N}$. 

Let us introduce an  operator ${\boldsymbol S} \in \mathbb{C}^{N \times N}$, where the $(i,j)$th entry of ${\boldsymbol S}$ denoted by $s_{i,j}$ { is nonzero only if $(i,j) \in \mathcal{E}$ and $s_{i,j}$ can also be nonzero if $i=j$ for $(i,j) \in \mathcal{E}$}, and is zero otherwise. The pattern of ${\boldsymbol S}$ captures the local structure of the graph. More specifically, for a graph signal $\x$, the signal 
$\Sb\x$ denotes the unit shifted version of $\x$. Hence ${\Sb}$ is referred to as the {\it graph-shift} operator~\cite{sandryhaila2013discrete}. Different choices for ${\boldsymbol S}$ include the graph Laplacian ${\boldsymbol L}$~\cite{shuman2013Emerging}, the adjacency matrix ${\boldsymbol A}$~\cite{sandryhaila2013discrete}, or their respective variants. {For undirected graphs, ${\boldsymbol S}$ is symmetric} (more generally, Hermitian), and thus it admits the following eigenvalue decomposition 
\begin{equation}
\label{eq:graphShift} 
\begin{aligned}
{\boldsymbol S} &= {\boldsymbol U}{\boldsymbol \Lambda}{\boldsymbol U}^H \\
&= [{\boldsymbol u}_1, \cdots, {\boldsymbol u}_N]\, {\rm diag}[\lambda_1,\cdots,\lambda_N] \, [{\boldsymbol u}_1, \cdots, {\boldsymbol u}_N]^H,
\end{aligned}
\end{equation}
where the eigenvectors $\{{\boldsymbol u}_n\}_{n=1}^N$ and the eigenvalues $\{\lambda_n\}_{n=1}^N$ of ${\boldsymbol S}$ provide the notion of frequency in the graph setting~\cite{shuman2013Emerging,sandryhaila2014big}. Specifically, $\{{\boldsymbol u}_n\}_{n=1}^N$ forms an orthonormal Fourier-like basis for graph signals with the graph frequencies denoted by $\{\lambda_n\}_{n=1}^N$.
%~\cite{shuman2013Emerging,sandryhaila2014big}. 
%(Jordan form if $\Sb$ is not diagonalizable~\cite{sandryhaila2013discrete})
Hence, the {\it graph Fourier transform} of a graph signal, $\x_f = [x_{f,1},x_{f,2},\ldots,x_{f,N}]^T \in \mathbb{C}^N$, is given by
\begin{equation}
\label{eq:graphFourier}
\x_f :={\boldsymbol U}^H \x \, \Leftrightarrow \, \x =: {\boldsymbol U} \x_f.
\end{equation} 

The frequency content of graph signals can be modified using {\it linear shift-invariant graph filters}~\cite{sandryhaila2013discrete,shuman2013Emerging}. { Let us call the system $\H \in \mathbb{C}^{N \times N}$ as a graph filter. If the eigenvalues of $\Sb$ are distinct, a shift-invariant graph filter, which satisfies ${\boldsymbol H}(\Sb\x) =\Sb({\boldsymbol H}\x)$, can be expressed as a polynomial in ${\boldsymbol S}$ as~\cite{sandryhaila2013discrete}}
\begin{equation}
\label{eq:graph_fitler}
\begin{aligned}
{\boldsymbol H} &=  h_0{\boldsymbol I} + h_1\Sb + \cdots + h_{L-1}\Sb^{L-1} \\
&= {\boldsymbol U}\left[h_0{\boldsymbol I}+ h_1{\boldsymbol \Lambda}+  \cdots + h_{L-1}{\boldsymbol \Lambda}^{L-1}\right]{\boldsymbol U}^H,
\end{aligned}
\end{equation}
where the filter ${\boldsymbol H}$ is of degree $L-1$ with filter coefficients ${\boldsymbol h} = [h_0,h_1,\ldots,h_{L-1}]^T \in \mathbb{C}^L$, and $L \leq N$ as $N$ is the degree of the minimal polynomial ({equal to the characteristic polynomial}) of $\Sb$. The diagonal matrix 
\begin{equation}
{\boldsymbol H}_f = \sum_{l=0}^{L-1} h_l{\boldsymbol \Lambda}^l = {\rm diag}[{\boldsymbol V}_L{\boldsymbol h}] = {\rm diag}[h_{f,1},\cdots,h_{f,N}] 
\label{eq:graphfilter_freq}
\end{equation}
can be viewed as the frequency response of the graph filter. Here, ${\boldsymbol V}_L$ is an $N \times L$ Vandermonde matrix with the $(i,j)$th entry as $\lambda_i^{j-1}$. 
%{We also assume that all linear graph filters \eqref{eq:graph_fitler} are shift invariant, i.e., ${\boldsymbol H}(\Sb\x) =\Sb({\boldsymbol H}\x)$, which requires all the eigenvalues of $\Sb$ to be distinct~\cite{sandryhaila2013discrete}.} 

%Finally, to ensure numerical stability, we will use a normalized graph-shift operator~\cite{sandryhaila2014discrete}, normalized to the maximum eigenvalue of $\Sb$, $\lambda_{\rm max}$, as $\Sb_{\rm norm} = \lambda_{\rm max}^{-1}\Sb$ so that $\| \Sb_{\rm norm} \x\|_2/ \|\x\|_2 \leq 1$.
%However, for conciseness, we will simply write  $\Sb_{\rm norm}$ as $\Sb$ from now on.
%That is, the vector $\tilde{\boldsymbol h} = {\boldsymbol V}{\boldsymbol h} \in \mathbb{R}^N$ contains the frequency responses of the filter.
%

\subsection{Stationary graph signals}

Let ${\boldsymbol x} = [x_1,x_2,\cdots,x_N]^T \in \mathbb{C}^N$ be a stochastic signal defined on the vertices of the graph $\mathcal{G}$ with expected value ${\boldsymbol m}_\x = \mathbb{E}\{{\boldsymbol x}\}$ and covariance matrix ${\boldsymbol R}_{\boldsymbol x} = \mathbb{E}\{({\boldsymbol x} - {\boldsymbol m}_\x)({\boldsymbol x} - {\boldsymbol m}_\x)^{H}\}$. Efforts to generalize some of the concepts of statistical time invariance or stationarity of signals defined over regular structures to random graph signals have been made in~\cite{benjamin15eusipco,perraudin2016stationary,marques2016stationary,girault2015signal}. For the sake of completeness, we will summarize the definitions from~\cite{benjamin15eusipco,perraudin2016stationary,marques2016stationary,girault2015signal} as follows.  
%girault2015translation
\begin{mydef}[Second-order stationarity] 

A random graph signal ${\boldsymbol x}$ is second-order stationary, if and only if, the following properties hold:
\begin{itemize}
%\item[1.] The mean of the graph signal is constant, $\mathbb{E}\{x_i\}$.
\item[1.] The mean of the graph signal is { collinear to an eigenvector of $\Sb$ corresponding to the smallest eigenvalue, i.e., ${\boldsymbol m}_\x = m_\x \u_1$}.
\item[2.] Matrices ${\boldsymbol S}$ and ${\boldsymbol R}_{\boldsymbol x}$ can be simultaneously diagonalized. 
\end{itemize}
\label{def:graphstationarity}
\end{mydef}
{ Since we assume that the eigenvalues of $\Sb$ are distinct and $\U$ forms an orthonormal basis, property 2 in the above definition essentially means the {\it statistical orthogonality} of spectral components, i.e,. $\mathbb{E}\{x_{f,i}{\bar{x}_{f,j}}\}=0$ for $i \neq j$~\cite{girault2015signal}.}

For simplicity, from now on we will focus on graph signals with zero mean, {where we assume that $m_\x$ is either known or $m_\x$ can be set to zero by preprocessing the data as discussed in Section~\ref{sec:numericalresults}}. We can generate zero-mean second-order stationary graph signals by graph filtering zero-mean white noise. Let ${\boldsymbol n} =[n_1,n_2,\ldots,n_N]^T \in \mathbb{C}^N$ be zero-mean unit-variance noise with covariance matrix ${\boldsymbol R}_{\boldsymbol n} = {\boldsymbol I}$. Then, a zero-mean second-order stationary graph signal ${\boldsymbol x}$ can be modeled as
$
{\boldsymbol x} = {\boldsymbol H}{\boldsymbol n},
$ 
where ${\boldsymbol H}$ can be any valid graph filter. The filtered signal will have zero mean and covariance matrix ${\boldsymbol R}_{\boldsymbol x} = \mathbb{E}\{({\boldsymbol H}{\boldsymbol n})({\boldsymbol H}{\boldsymbol n})^H\}$ given by
\begin{equation}
\label{eq:cov_filteredsignal}
\begin{aligned}
{\boldsymbol R}_{\boldsymbol x}& = {\boldsymbol H}{\boldsymbol R}_{\boldsymbol n}{\boldsymbol H}^H   \\
&= {\boldsymbol U} {\rm diag}[|h_{f,1}|^2,\cdots,|h_{f,N}|^2]  {\boldsymbol U}^H\\
&= {\boldsymbol U} {\rm diag}[{\boldsymbol p}] {\boldsymbol U}^H,
\end{aligned}
\end{equation}
where $h_{f,n} = h_0 + h_1\lambda_n + \cdots + h_{L-1}\lambda_n^{L-1}$ is defined in \eqref{eq:graphfilter_freq}. This conforms to the second property listed in Definition~\ref{def:graphstationarity}. More generally, graph filtering any second-order stationary graph {signal} also results in a second-order stationary graph {signal} (it is easy to verify this using property $2$ in Definition~\ref{def:graphstationarity}).  The {nonnegative vector} ${\rm diag}[{\boldsymbol p}]$ in \eqref{eq:cov_filteredsignal} is referred to as the {\it graph power spectral density} or {\it graph power spectrum}. We now formally introduce the graph power spectrum through the following definition.

\begin{mydef}[Graph power spectrum]
The graph power spectral density of a second-order stationary graph signal is a real-valued nonnegative length-$N$ vector ${\boldsymbol p} = [p_1,p_2,\ldots,p_N]^T \in \mathbb{R}_+^N$ with entries given by
\begin{equation}
\label{eq:graphPSD}
p_n = {\boldsymbol u}_n^H {\boldsymbol R}_{\boldsymbol x} {\boldsymbol u}_n, \, n=1,2,\ldots,N.
\end{equation}
Alternatively, $p_n = |h_{f,n}|^2 \geq 0$, for $n=1,2,\ldots,N$, where  $h_{f,n} = h_0 + h_1\lambda_n + \cdots + h_{L-1}\lambda_n^{L-1}$ is defined in \eqref{eq:graphfilter_freq}.
\end{mydef}
%An example of a second-order stationary graph signal is white noise with zero mean (or with a constant mean) and covariance matrix ${\boldsymbol R}_{\boldsymbol x} = \sigma^2{\boldsymbol I}$, which can always be expressed as ${\boldsymbol R}_{\boldsymbol x} = {\boldsymbol U}{\rm diag}[{\boldsymbol p}]{\boldsymbol U}^H$ with a flat graph power spectrum, i.e.,  $p_1=p_2=\cdots = p_N=  \sigma^2$.

Second-order stationarity is preserved by linear graph filtering. This means that stationary graph signals with a prescribed graph power spectrum can be generated by filtering white noise, where the graph power spectrum of the filtered signal is reshaped according to the frequency response of the graph filter~\cite{benjamin15eusipco,perraudin2016stationary,marques2016stationary}. As a result, the graph power spectrum reveals critical information about the second-order stationary graph signal, and thus estimating the graph power spectrum or recovering the second-order statistics of a graph signal is useful in many applications. 

{ We end this section by summarizing the list of assumptions made in this paper. 
\begin{enumerate}
\item The shift operator $\Sb$ is known.
\item The orthonormal basis $\U$ and the distinct eigenvalues  $\{\lambda_n\}_{n=1}^N$ of $\Sb$ are known a priori. 
\end{enumerate}
}

\section{Non-parametric Spectral Domain Approach} \label{sec:subsampling}

The size of the datasets {often} inhibits a direct computation of the second-order statistics, e.g., { by observing all the $N$ nodes and using \eqref{eq:graphPSD} to compute the graph power spectrum. This would computationally cost $\mathcal{O}(N^3)$}.
%
%
%by computing the sample data covariance matrix, { based on multiple snapshots, say $N_s$, by observing all the $N$ nodes. Forming the sample covariance matrix computationally costs $\mathcal{O}(N_sN)$ and computing the graph power spectrum using \eqref{eq:graphPSD} costs $\mathcal{O}(N^3)$.}
%
%or by computing the graph power spectrum using \eqref{eq:graphPSD}, which requires diagonalization of the graph-shift operator that computationally costs $\mathcal{O}(N^3)$. 
As such, compression or data reduction is preferred especially for large-scale data in the graph setting{~\cite{sandryhaila2014big}}.  In the context of graph signal processing, most works consider subsampling the graph signal $\x$ assuming some spectral prior to reconstruct it~\cite{anis2014towards,tsitsvero2015uncertainty,varma2015spectrum,chen2015discrete}. This approach is, in principle, also possible for recovering the second-order statistics of $\x$. However, when the goal is to reconstruct the second-order statistics of $\x$ (and not $\x$ itself), it is computationally advantageous, and allows for a stronger compression, when we avoid the intermediate step of reconstructing and storing $\x$. In this paper, we will therefore focus on recovering graph second-order statistics directly from subsampled graph signals. We refer to this problem as {\it graph covariance subsampling}.

The extension of compressive covariance sensing~\cite{ariananda2012compressive,romero2013wideband,Romero16CCSspm} to graph covariance subsampling is non-trivial. This is because for second-order (or wide-sense) stationary signals with a regular support, the covariance matrix has a clear structure (e.g., Toeplitz, circulant) that enables an elegant subsampler design, but for second-order stationary graph signals residing on arbitrary graphs, the covariance matrix does not admit {\it any} clear structure that can be easily exploited, in general.

Consider the problem of estimating the graph power spectrum of the second-order stationary graph signal ${\boldsymbol x} \in \mathbb{C}^N$ from a set of $K \ll N$ linear observations stacked in the vector ${\boldsymbol y} \in \mathbb{C}^K$, given by
\begin{equation}
\label{eq:subsampling}
{\boldsymbol y} = {\boldsymbol \Phi}{\boldsymbol x}, % = {\boldsymbol \Phi}{\boldsymbol H} {\boldsymbol n}, 
\end{equation}
where ${\boldsymbol \Phi}$ is a known $K \times N$ selection matrix with Boolean entries, i.e.,  ${\boldsymbol \Phi} \in \{0,1\}^{K \times M}$ (we will discuss the subsampler design in Section~\ref{sec:subsampdesgn}) and where several realizations of $\y$ may be available. 
%For simplicity, we assume in this paper that ${\boldsymbol x}$ is zero mean.
%and that it is obtained by graph filtering zero-mean white noise ${\boldsymbol n}$ with the filter ${\boldsymbol H}$. 
The matrix ${\boldsymbol \Phi}$ is referred to as the {\it subsampling} or {\it sparse sampling} matrix, where the compression is achieved by setting~$K \ll N$. { For applications where graph nodes correspond to sensing devices (e.g., weather stations in climatology, electroencephalography (EEG) probes in brain networks), such a sparse sampling scheme results in a significant reduction in the hardware, storage and communications costs next to the reduction in the processing costs.} 

The covariance matrices $\R_\x = \mathbb{E}\{\x\x^H\} \in \mathbb{C}^{N \times N}$ and $\R_\y = \mathbb{E}\{\y\y^H\}\in \mathbb{C}^{K \times K}$ contain the second-order statistics of $\x$ and $\y$, respectively. { In practice, typically, multiple snapshots, say $N_s$ snapshots, are observed to form a sample covariance matrix. Forming the sample covariance matrix from $N_s$ snapshots of $\x$ costs $\mathcal{O}(N^2N_s)$, while forming the sample covariance matrix from $N_s$ snapshots of $\y$ {\it only} costs $\mathcal{O}(K^2N_s)$.} We now state the problem of interest as follows.
\begin{myprob} ({\bf Recovering second-order statistics}) For a known undirected graph $\mathcal{G}$, given {a number of realizations , say $N_s$,} of the subsampled length-$K$ graph signal $\y$ or the subsampled covariance matrix $\R_\y$, recover the { graph power spectrum~$\p$} and thus the covariance matrix $\R_\x$.
\end{myprob}

Let us decompose the graph signal $\x$ in terms of its graph Fourier transform coefficients as [cf. \eqref{eq:graphFourier}]
\[
\x  = \sum_{i=1}^N x_{f,i} \u_i. % = \sum_{i=1}^N n_ih_{f,i} \u_i.
\]
This allows us to represent the covariance matrix $\R_\x= \mathbb{E}\{\x\x^H\}$ in the graph Fourier domain using the graph power spectrum $\p$ as
\begin{align}
\label{eq:Rxspectral}
\R_\x = \sum_{i=1}^N \mathbb{E}\{|x_{f,i}|^2\} \u_i\u_i^H = \sum_{i=1}^N p_i \u_i\u_i^H = \sum_{i=1}^N p_i \Q_i,
\end{align}
where we use the fact that for $i \neq j$ we have $\mathbb{E}\{x_{f,i}\bar{x}_{f,j}\}  =0$ and $\Q_i = \u_i\u_i^H$ is a size-$N$ rank-one matrix. Here, we expand $\R_\x$ using a set of $N$ Hermitian matrices $ \{\Q_1,\Q_2,\ldots,\Q_N\}$ as a basis. 
 %and recall that $\u_i \neq h$
Vectorizing $\R_x$ in \eqref{eq:Rxspectral} results in
\[
\r_\x = {\rm vec}(\R_\x) = \sum_{i=1}^N p_i {\rm vec}(\Q_i)= \Psib_{\rm s} \p, %\sum_{i=1}^N p_i \Q_i
\]
where we have stacked ${\rm vec}(\Q_i) = \bar{\u}_i \otimes \u_i$ to form the $N^2 \times N$ matrix 
$\Psib_{\rm s}$ as
\[
\Psib_{\rm s} = [\bar{\u}_1 \otimes \u_1,\cdots,\bar{\u}_N \otimes \u_N] = \bar{\U} \circ \U.
\]
The subscript ``${\rm s}$" in the matrix ${\boldsymbol \Psi}_{\rm s}$, which is constructed using the graph Fourier basis vectors, stands for {\it spectral domain}. 

Using the compression scheme described in \eqref{eq:subsampling}, the covariance matrix $\R_\y \in \mathbb{C}^{K \times K}$ of the subsampled graph signal ${\boldsymbol y}$ can be related to $\R_\x$ as  
\begin{equation}
\label{eq:subsamCov}
\begin{aligned}
{\boldsymbol R}_{\boldsymbol y} = {\boldsymbol \Phi}{\boldsymbol R}_{\boldsymbol x}{\boldsymbol \Phi}^T  =  \sum_{i=1}^N p_i \Phib \Q_i \Phib^T.
\end{aligned}
\end{equation}
This means that the expansion coefficients of $\R_\y$ with respect to the set $\{\Phib\Q_1\Phib^T,\Phib\Q_2\Phib^T,\cdots,\Phib\Q_N\Phib^T\}$ are the {\it same} as those of $\R_\x$ with respect to the set $\{\Q_1,\Q_2,\cdots,\Q_N\}$, and they are preserved under linear compression. It is not yet clear though whether these expansion coefficients, which basically represent the power spectrum, can be uniquely recovered from~$\R_\y$.

Vectorizing $\R_\y$ as \[
{\boldsymbol r}_{\boldsymbol y} = {\rm vec}({\boldsymbol R}_{\boldsymbol y}) = ({\boldsymbol \Phi} \otimes {\boldsymbol \Phi}) {\rm vec}(\R_\x) \in \mathbb{C}^{K^2}
\]
 we obtain
\begin{equation}
\label{eq:subsamCovVec}
\begin{aligned}
{\boldsymbol r}_{\boldsymbol y} &= \sum_{i=1}^N p_i ({\boldsymbol \Phi} \otimes {\boldsymbol \Phi}) (\bar{\u}_i \otimes \u_i) = \sum_{i=1}^N p_i ({\boldsymbol \Phi}\bar{\u}_i \otimes {\boldsymbol \Phi}\u_i) \\
&= ({\boldsymbol \Phi} \otimes {\boldsymbol \Phi}) {\boldsymbol \Psi}_{\rm s}{\boldsymbol p}.
\end{aligned}
\end{equation}
This linear system with $N$ unknowns has a unique solution if $({\boldsymbol \Phi} \otimes {\boldsymbol \Phi}) {\boldsymbol \Psi}_{\rm s}$ has full column rank, which requires $K^2 \geq N$. Assuming that this is the case, the  graph power spectrum (thus the second-order statistics of $\x$) can be estimated in closed form via least squares:
\begin{equation}
\label{eq:ls_gspectral}
\widehat{\boldsymbol p} = [({\boldsymbol \Phi} \otimes {\boldsymbol \Phi}) {\boldsymbol \Psi}_{\rm s}]^\dagger {\boldsymbol r}_{\boldsymbol y}.
\end{equation}
{Computing this least squares solution costs $\mathcal{O}(K^2N^2)$~\cite{golub1996matrix}. Although for the non-parametric approach, cost of computing \eqref{eq:ls_gspectral} is on the same order as that of the uncompressed case, the cost reduction will be prominent for problems discussed later on in Section~\ref{sec:parametric}. Further, to compute \eqref{eq:ls_gspectral}, we have assumed that the true covariance matrix $\R_\y$ is available, but a practical scenario with finite data records is discussed in Section~\ref{sec:finitedata}.} 

%For the above least-squares solution to be unique, it is necessary that the matrix $({\boldsymbol \Phi} \otimes {\boldsymbol \Phi}) {\boldsymbol \Psi}_{\rm s}$ is tall and it has full column rank. 
% we need $K^2 \geq N$ as show next: 
%As we will show next, the subsampling matrix $\Phib$ plays a crucial role in enabling the 
\begin{mydef} A wide matrix $\Phib$ is a valid graph covariance subsampler if it yields a full column rank matrix 
$({\boldsymbol \Phi} \otimes {\boldsymbol \Phi}) {\boldsymbol \Psi}_{\rm s}$.
\end{mydef}
We now derive the conditions under which ${\boldsymbol \Phi}$ is a valid graph covariance subsampler. To do this, we first introduce two important lemmas.
\begin{mylem} 
\label{lem:Psifullrank}
Since the matrix $\U \in \mathbb{C}^{N \times N}$ is full rank, the matrix $\Psib_{\rm s} = \bar{\U} \circ \U$ of size ${N^2 \times N}$ has full column rank.
\end{mylem}

\begin{proof}
See Appendix~\ref{app:proofselfkr}.
\end{proof}

\begin{mylem}
\label{lem:rankKron}
 If the matrix $\Phib \in \mathbb{R}^{K \times N}$ has full row rank, then the matrix $ \Phib \otimes \Phib$ of size $K^2 \times N^2$ also has full row rank.
\end{mylem}

\begin{proof}
%See Appendix~\ref{app:proofselfkron}.
Follows from the singular value decomposition of $\Phib$ and the property $({\boldsymbol A} \otimes {\boldsymbol B}) ({\boldsymbol C} \otimes  {\boldsymbol D})  = ({\boldsymbol A}{\boldsymbol C} \otimes {\boldsymbol B}{\boldsymbol D})$.
\end{proof}

Using the above two lemmas, we can provide the necessary and sufficient conditions under which the solution in \eqref{eq:ls_gspectral} is unique.
\begin{mytheo}
\label{theo:fullrank_spectral}
A full row rank matrix $\Phib \in \mathbb{R}^{K \times N}$ is a valid graph covariance subsampler if and only if
the matrix $({\boldsymbol \Phi} \otimes {\boldsymbol \Phi}) {\boldsymbol \Psi}_{\rm s}$ is tall, i.e., $K^2 \geq N$, 
and ${\rm null}(\Phib \otimes \Phib) \cap {\rm ran}(\Psib_s) = \emptyset$. 

%If the matrix $\Phib \in \mathbb{R}^{K \times N}$ has full row rank and if $K^2 \geq N$, then the matrix $({\boldsymbol \Phi} \otimes {\boldsymbol \Phi}) {\boldsymbol \Psi}_{\rm s}$ will have full column rank.  
\end{mytheo}

\begin{proof}
See Appendix~\ref{app:theo1}.
\end{proof}

Although the linear system of equations \eqref{eq:subsamCovVec} can be solved using (unconstrained) least squares, {nonnegativity constraints or} any spectral prior can be easily accounted for while solving \eqref{eq:subsamCovVec} as summarized in the following remark. 

\begin{myrem}[Spectral priors] 
\label{rem:cls_prior}
Any available prior information about the graph spectrum might allow for a higher compression with $K^2 < N$, or an improvement of the solution \eqref{eq:ls_gspectral}. Suppose we have some prior knowledge about the graph spectrum, i.e., $\p \in \mathcal{P}$ with $\mathcal{P}$ being the constraint set. For instance, suppose we know a priori that (a) the spectrum is bandlimited (e.g., lowpass) with known support such that $\mathcal{P} = \{ \p \, |\, p_n = 0, n \notin [N_l, N_u] \}$, where $[N_l,N_u]$ denotes the support set, (b) the spectrum is sparse, but with unknown support such that $\mathcal{P} := \{ \p \, | \, \sum_{n=1}^{N} p_n = S\}$, where $S$ denotes the sparsity order {(here, we use the convex relaxation of the cardinality constraint), or (c) the power spectrum is nonnegative (by definition), for which $\mathcal{P} := \{ \p \, | \, p_n \geq0, \forall n\}$.} With such spectral priors, the following constrained least squares problem may be solved
\[
\underset{\p \in \mathcal{P}}{\text{\rm minimize}} \quad \| \r_\y - (\Phib \otimes \Phib) \Psib_{\rm s} \p\|_2^2.
\]
\end{myrem}

In what follows, we will discuss and illustrate the connections with compressive covariance sensing~\cite{ariananda2012compressive,Romero16CCSspm} for datasets that reside on regular structures (e.g., time series) using a circulant graph (e.g., a cycle graph). We will also see that designing a compression matrix is much more elegant for such circulant graphs.

\section{Circulant Graphs} \label{sec:cyclicgraphs}

Discrete-time finite or periodic data can be represented using \emph{directed} cycle graphs, where the direction of the edge represents the evolution of time from past to future. The edge directions may be ignored {in some cases, e.g., when we are only interested in exploiting the regular Fourier transform, when we are dealing with the spatial domain, or when the underlying data is a time-reversible stochastic process that is invariant under the reversal of the time scale~\cite{weiss1975time}}. 
In such cases, the data can be represented using an \emph{undirected} cycle graph, see Fig.~\ref{fig:cyclegraph}. 
\begin{figure}[!t]
\centering
\input{figures/cycle_graph.tex}
\caption{Undirected cycle graph. The graph covariance matrix of stationary signals $\{x_n\}_{n=1}^N$ supported on this undirected cycle graph will be a circulant matrix.}
\label{fig:cyclegraph}
\end{figure}
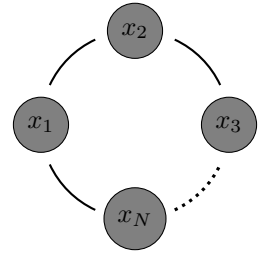

Consider the adjacency matrix of this undirected cycle graph as its graph-shift operator, which 
%\[
%\Sb =\left[\begin{array}{cccc} & 1 &  & 1 \\1 &  & \ddots &  \\ & \ddots &  & 1 \\1 &  & 1 & \end{array}\right]. %\left[\begin{array}{cccc}0 & 1 & 0 & 1 \\1 & 0 & 1 & 0 \\0 & 1 & 0 & 1 \\1 & 0 & 1 & 0\end{array}\right].
%\]
%This 
will be an $N \times N$ symmetric circulant matrix. We know that a circulant matrix can be diagonalized with a discrete Fourier transform matrix. In other words, the graph Fourier transform matrix $\U$ related to this graph will consist of the orthonormal vectors 
\[
\u_n = [\om_n^0, \om_n, \om_n^2,\cdots,\om_n^{N-1}]^T
\]
with $\om_n = \exp(-\imath2\pi n/N)/\sqrt{N}$ and it will be a Vandermonde matrix {(here, $\imath^2 = -1$)}. In general, for {\it circulant graphs} with circulant graph-shift operators, an eigenvalue decomposition is not required to compute the graph Fourier transform matrix $\U$ or the model matrix $\Psib_{\rm s}$, which was introduced in Section~\ref{sec:subsampling}.

%The $n$th column of $\Psib_{\rm s}$, i.e., $\bar{\u}_n \otimes \u_n$ contains the elements $\{\om_n^m\}$ for $m = -(N-1),\cdots,0,\cdots,(N-1)$ (with a few repeated entries). The increase in the frequency range (i.e., to $n= -(N-1),\ldots,0$) is due to the differences $\{n_i-n_j\}$ for $0\leq i \leq N-1$, $0\leq j \leq N-1$. 

Let the set $\mathcal{K} \subset \mathcal{N}$ denote the indices of the selected graph nodes. Now, if we can smartly select the entries of ${\u}_n$ such that the related entries of $\bar{\u}_n \otimes \u_n$ contain all the distinct values $\{\om_n^m\}$ for $m = 0,\cdots,N-1$, the matrix $(\Phib \otimes \Phib)\Psib_{\rm s}$ will be a full-column rank Vandermonde matrix.  In particular, this means that, for every $m=0,\ldots,N-1$, there must exist at least one pair of elements $n_i, n_j \in \mathcal{K}$ that satisfies $n_i - n_j = m$, { where the difference $n_i-n_j$ is due to the Kronecker product $\bar{\u}_n \otimes \u_n$}.  Sets $\mathcal{K}$ having this property are called {\it sparse rulers}~\cite{Romero16CCSspm}. Furthermore, if the set contains a minimum number of elements, they are called \emph{minimal sparse rulers}, which results in the best possible compression. 

Let us illustrate this with an example for $N=10$. In this case, the set $\mathcal{K} = \{0,1,4,7,9\}$ with $K = |\mathcal{K}| = 5$ elements is a minimal sparse ruler. In other words, by choosing the subsampling matrix $\Phib = \diag_{\rm r}[{\boldsymbol w}]$ with ${\boldsymbol w} = [1, 1, 0, 0, 1, 0, 0, 1, 0, 1]^T$
%\[
%\Phib = \left[\begin{array}{cccccccccc}1 & 0 & 0 & 0 & 0 & 0 & 0 & 0 & 0 & 0 \\0 & 1 & 0 & 0 & 0 & 0 & 0 & 0 & 0 & 0 \\0 & 0 & 0 & 0 & 1 & 0 & 0 & 0 & 0 & 0 \\0 & 0 & 0 & 0 & 0 & 0 & 0 & 1 & 0 & 0 \\0 & 0 & 0 & 0 & 0 & 0 & 0 & 0 & 0 & 1\end{array}\right],
%\]
we can ensure that the matrix $({\boldsymbol \Phi} \otimes {\boldsymbol \Phi}) {\boldsymbol \Psi}_{\rm s}$ is full column rank, and hence the second-order statistics of $\x$ can be estimated using \eqref{eq:ls_gspectral} by subsampling only $K=5$ nodes. Here, we achieve a compression rate of $K/N = 0.5$. Similarly, for $N=80$, the minimal sparse ruler has $K=15$ elements, and this results in a compression rate of $K/N=0.1875$ {(we will see an example related to $N=80$ and $K=15$ in Section~\ref{sec:numericalresults})}. Sparse rulers for other values of $N$ are tabulated in~\cite{linebarger1993difference}. 

Computing minimal sparse rulers is a combinatorial problem with no known expressions. Nevertheless, subsamplers such as coprime~\cite{vaidyanathan2011sparse} and nested sparse samplers~\cite{pal2010nested}, which can be computed using a closed-form expression for any $N$, are also valid covariance subsamplers. However, they are not minimal sparse rulers and thus they do not provide the best compression rate.

Subsampler design for reconstructing the second-order statistics of signals residing on a circulant graph is as elegant as that for reconstructing the second-order statistics of stationary time-series. The design of subsamplers for general graphs, however, is more challenging. This is the subject of Section~\ref{sec:subsampdesgn}.

\section{Parameteric Models} \label{sec:parametric}

In this section, we will focus on a parametric representation of the graph power spectrum. In particular, the focus will be on \emph{moving average} and \emph{autoregressive} parametric models. Typically, the model order (i.e., the number of parameters) is much smaller than the length of the graph signal, and since we now have to recover only these parameters, a much stronger compression can be achieved. {Also, this means that, we need to store or transmit only fewer parameters, which could be used to generate realizations of second-order stationary graph signals (we will illustrate this with an example in Section~\ref{sec:numericalresults})}

Parametric methods can be viewed as an alternative approach, where going to the graph spectral domain may be avoided, and instead, all the processing is done directly in the graph vertex domain. 
%Hence, computing an expensive eigenvalue decomposition of the graph shift matrix can be avoided. 

\subsection{Graph moving average models}\label{sec:GraphVertexDomain}

As before, we assume that the stationary graph signal $\x$ is generated by graph filtering zero-mean {unit-variance} white noise. { Recall that  in Section~\ref{sec:subsampling}, we did not impose any structure to the graph filter, but now we will assume that the graph filter has a finite impulse response with an all-zero form as in \eqref{eq:graph_fitler}; see~\cite{perraudin2016stationary,marques2016stationary}.} 
%As a result, the signal $\x$ can be seen as a \emph{moving average graph process}. 

Let us begin by writing the graph signal $\x$ as
\[
\x = \H(\h)\n=\sum_{l=0}^{L-1} h_l \Sb^l \n = {\boldsymbol U}\left(\sum_{l=0}^{L-1} h_l{\boldsymbol \Lambda}^l\right){\boldsymbol U}^H\n
\]
with covariance matrix
\begin{equation}
\begin{aligned}
\label{eq:RxMA}
\R_\x &= {\H}(\h){\H}^H(\h) \\ 
&= {\boldsymbol U}\left(\sum_{l=0}^{L-1}  h_l{\boldsymbol \Lambda}^l\right) \left(\sum_{l=0}^{L-1}  \bar{h}_l{\boldsymbol \Lambda}^l\right){\boldsymbol U}^H,
\end{aligned}
\end{equation}
where $\x$ is a {\it moving average graph {signal}} ({\it G-MA}) of order $L-1$ with {\it G-MA} coefficients $\{h_k\}_{k=0}^{L-1}$, and the length-$L$ vector $\h$ collects the {\it G-MA} coefficients as $\h = [h_0, h_1,\ldots,h_{L-1}]^T$. { Moving average models are particularly useful to represent a smooth graph power spectrum~\cite{perraudin2016stationary,marques2016stationary}.}
%A graph signal that satisfies the above moving average model is referred to as a {moving average graph process} of order $L-1$. 

The expression \eqref{eq:RxMA} basically means that we can express the covariance matrix ${\boldsymbol R}_{\boldsymbol x}$ as a polynomial of the graph shift operator:
\begin{equation}
\label{eq:Rxpolynom}
{\boldsymbol R}_{\boldsymbol x} =  \sum_{k=0}^{Q-1} b_k {\boldsymbol S}^k,
\end{equation}
where $Q=\min\{2L-1,N\}$ unknown expansion coefficients $\{b_k\}_{k=0}^{Q-1}$ collected in the vector ${\boldsymbol b} = [b_0,b_1,\cdots,b_{Q-1}]^T \in \mathbb{R}^Q$ completely characterize the covariance matrix $\R_\x$. In other words, we assume a 
linear parametrization of the covariance matrix ${\boldsymbol R}_{\boldsymbol x}$ using the set of $Q$ Hermitian matrices $\{{\boldsymbol S}^0,{\boldsymbol S},\cdots,{\boldsymbol S}^{Q-1}\}$ as a basis. 

{
The expansion coefficients $\b$ depend on the G-MA coefficients $\h$. To see this, let us consider an example {\it G-MA} model with $L=3$ having coefficients $\h = [h_0,h_1,h_2]^T$, for which \eqref{eq:Rxpolynom} simplifies to
\begin{align}
\R_\x =  h_0^2 {\boldsymbol I} &+ 2 h_0h_1 \Sb + (h_1^2 + 2h_0h_2) \Sb^2 \nonumber \\ 
&\>+ 2h_1h_2  \Sb^3 + h_2^2 \Sb^4. \label{eq:MAexample}
\end{align}
This means that, $\b(\h)$ will  be of length $2L-1$ with entries $\b(\h) = [h_0^2, 2h_0h_1,h_1^2 + 2h_2h_0,2h_2h_1, h_2^2]^T$ that are related to the G-MA parameters $\h$. To arrive a {\it simple} (unconstrained) least squares estimator, we will ignore this structure in $\b$ (we will discuss the how to account for this structure at the end of this subsection). Therefore, with a slight abuse of notation we will henceforth refer to $\b(\h)$ as the {\it G-MA} coefficients.    

}

Depending on the shape of the power spectrum, $Q$ can be much smaller than the number of graph nodes (i.e., the length of the vector ${\boldsymbol p}$) thus allowing a higher compression. In any case, the value of $Q$ will be at most $N$, recalling that $N$ is the degree of the minimal {(and characteristic)} polynomial of $\Sb$. That is to say, for $Q \geq N$, the set of matrices $\{{\boldsymbol S}^0,{\boldsymbol S},\cdots,{\boldsymbol S}^{Q-1}\}$ are linearly  dependent. 
%(The set of matrices $\{{\boldsymbol S}^0,{\boldsymbol S},\cdots,{\boldsymbol S}^{Q-1}\}$ are linearly independent if  $\sum_{k=0}^{Q-1} \alpha_k\Sb^q = \sum_{k=0}^{Q-1} \beta_k\Sb^q$ only if $ \alpha_k = \beta_k, \, \forall q$.)

Vectorizing ${\boldsymbol R}_{\boldsymbol x}$ in \eqref{eq:Rxpolynom} yields
\begin{equation}
\label{eq:RxpolynomVec}
{\boldsymbol r}_{\boldsymbol x} = {\rm vec}({\boldsymbol R}_{\boldsymbol x}) = \sum_{k=0}^{Q-1} b_k {\rm vec}({\boldsymbol S}^q) = {\boldsymbol \Psi}_{\rm MA}{\boldsymbol b},
\end{equation}
where we have stacked ${\rm vec}({\boldsymbol S}^q)$ to form the columns of the matrix ${\boldsymbol \Psi}_{\rm MA} \in \mathbb{R}^{N^2 \times Q}$ as
{
\[
{\boldsymbol \Psi}_{\rm MA} = \left[{\rm vec}({\boldsymbol S}^0),{\rm vec}({\boldsymbol S}^1),\cdots,{\rm vec}({\boldsymbol S}^{Q-1})\right],
\]
}
and the subscript ``${\rm MA}$" in ${\boldsymbol \Psi}_{\rm MA}$ stands for {\it moving average}.

The covariance matrix of the subsampled graph {signal} $\y$ in \eqref{eq:subsampling} will then be
\begin{equation}
\label{eq:SubRyMA}
{\boldsymbol R}_{\boldsymbol y} =  {\boldsymbol \Phi}{\boldsymbol R}_{\boldsymbol x}{\boldsymbol \Phi}^T = \sum_{k=0}^{Q-1} b_k \Phib\Sb^k\Phib^T.
\end{equation}
As in the graph spectral domain approach discussed in Section~\ref{sec:subsampling}, the G-MA coefficients $\{b_k\}_{k=0}^{Q-1}$ of $\R_\y$ with respect to the set $\{\Phib{\boldsymbol S}^0\Phib^T,\Phib{\boldsymbol S}\Phib^T,\cdots,\Phib{\boldsymbol S}^{Q-1}\Phib^T\}$ are the {\it same} as those of $\R_\x$ with respect to the set  $\{{\boldsymbol S}^0,{\boldsymbol S},\cdots,{\boldsymbol S}^{Q-1}\}$.

Vectorizing $\R_\y$, we get a set of $K^2$ equations in $Q$ unknowns, given by 
\begin{equation}
\label{eq:ry_vertex}
\begin{aligned}
{\boldsymbol r}_{\boldsymbol y} = {\rm vec}({\boldsymbol R}_{\boldsymbol y}) &= ({\boldsymbol \Phi} \otimes {\boldsymbol \Phi}) {\rm vec}({\boldsymbol R}_{\boldsymbol x})\\
&= ({\boldsymbol \Phi} \otimes {\boldsymbol \Phi}){\boldsymbol \Psi}_{\rm MA}{\boldsymbol b}. 
\end{aligned}
\end{equation}
If the matrix $({\boldsymbol \Phi} \otimes {\boldsymbol \Phi}){\boldsymbol \Psi}_{\rm MA}$ has full column rank, which requires $K^2 \geq Q$, then the overdetermined system \eqref{eq:ry_vertex} can be uniquely solved using least squares as
\begin{equation}
\label{eq:MAls}
\widehat{\boldsymbol b} = [({\boldsymbol \Phi} \otimes {\boldsymbol \Phi}){\boldsymbol \Psi}_{\rm MA}]^\dag {\boldsymbol r}_{\boldsymbol y}.
\end{equation}

\begin{mycor} A full row rank matrix $\Phib \in \mathbb{R}^{K \times N}$ is a valid graph covariance subsampler if and only if the matrix $({\boldsymbol \Phi} \otimes {\boldsymbol \Phi}) {\boldsymbol \Psi}_{\rm MA}$ is tall, i.e., $K^2 \geq Q$, and ${\rm null}(\Phib \otimes \Phib) \cap {\rm ran}(\Psib_{\rm MA}) = \emptyset$. 
\end{mycor}
\begin{proof}
Follows from Theorem~\ref{theo:fullrank_spectral}. 
%The fact that $\mathcal{S}$ is linearly independent implies that the matrix
%$\Psib_{\rm MA}$ has full column rank, i.e., ${\rm rank}(\Psib_{\rm MA}) = Q$. So, if $\Phib$ is carefully designed such that 
%${\rm null}(\Phib \otimes \Phib) \cap {\rm ran}(\Psib_{\rm MA}) = \emptyset$, then the tall matrix $({\boldsymbol \Phi} \otimes {\boldsymbol \Phi}) {\boldsymbol \Psi}_{\rm MA}$ will have full column rank $Q$.
\end{proof}
%
%
%Recall that the construction of $\Psib_{\rm s}$ in Section~\ref{sec:subsampling} requires the computation of the eigenvalue decomposition of $\Sb$. This computationally  costs $\mathcal{O}(N^3)$. In contrast, the matrix $\Psib_{\rm MA}$ is constructed using $\Sb$, which is computationally more attractive.
%
%
%\subsubsection{Computing $\p$ from $\b$}
Although knowing the moving average filter coefficients ${\boldsymbol b}$ is equivalent to knowing $\R_\x$, it might be interesting to study the relation between ${\boldsymbol b}$ and the power spectrum $\p$. We can relate the vector ${\boldsymbol p}$ and the vector ${\boldsymbol b}$, by using \eqref{eq:graphPSD} and \eqref{eq:Rxpolynom}. That is, we can write $p_n = \sum_{k=0}^{Q-1} b_k {\lambda}_n^k$, or in matrix-vector form we have 
\[
{\boldsymbol p} = {\boldsymbol V}_Q {\boldsymbol b},
\] where ${\boldsymbol V}_Q$ is an $N \times Q$ Vandermonde matrix with $(i,j)$th entry equal to $\lambda_i^{j-1}$. To recover ${\boldsymbol p}$ from ${\boldsymbol b}$, however, we need all the $N$ eigenvalues of $\Sb$ to construct ${\boldsymbol V}_Q$. 

This relation between ${\boldsymbol p}$ and ${\boldsymbol b}$ can be used to show the equivalence between the linear models \eqref{eq:subsamCovVec} and \eqref{eq:ry_vertex} as follows. The fact that ${\boldsymbol S}^q = {\boldsymbol U}{\boldsymbol \Lambda}^q{\boldsymbol U}^H$ from \eqref{eq:graphShift} allows us to express 
${\boldsymbol \Psi}_{\rm MA}$ in \eqref{eq:ry_vertex} as ${\boldsymbol \Psi}_{\rm MA} = ({\bar{\boldsymbol U}} \circ  {\boldsymbol U}){\boldsymbol V}_Q.$
%\begin{equation}
%\label{eq:PsivecS}
%{\boldsymbol \Psi}_{\rm MA} = ({\bar{\boldsymbol U}} \circ  {\boldsymbol U}){\boldsymbol V}_Q.
%\end{equation}
%Using \eqref{eq:PsivecS} in \eqref{eq:ry_vertex}, we obtain
Using this in \eqref{eq:ry_vertex}, we obtain
%\begin{equation}
%\label{eq:ry_vertexSim}
%\begin{aligned}
${\boldsymbol r}_{\boldsymbol y} = ({\boldsymbol \Phi} \otimes {\boldsymbol \Phi}) ({\bar{\boldsymbol U}} \circ  {\boldsymbol U}) {\boldsymbol V}_Q {\boldsymbol b} = ({\boldsymbol \Phi}{\bar{\boldsymbol U}} \circ {\boldsymbol \Phi}{\boldsymbol U})  {\boldsymbol p}.$
%\end{aligned}
%\end{equation}
%\end{myrem}
%
%Although we refer to the coefficients $\b$ as the G-MA coefficients, the dependence of $\b$ on $\h$ was ignored in the (unconstrained) least squares estimator \eqref{eq:MAls}. 

In the following, we exploit the structure in $\b$, which we ignored while solving \eqref{eq:ry_vertex}, to develop a constrained least squares estimator.

{
\begin{myrem}[Constrained least squares]
\label{rem:cls_ma}
To reveal the structure in $\b(\h)$, let us recall the example \eqref{eq:MAexample} with $L=3$. The coefficients in $\b(\h)$ are related to the squared polynomial $p(t) = (h_0 + h_1t+h_2t^2)^2$, which can also be written as
\begin{equation*}
%\label{eq:sqrdpol}
\begin{aligned}
p(t) &= \h^T \left[\begin{array}{ccc}1 & t & t^2 \\t & t^2 & t^3 \\t^2 & t^3 & t^4\end{array}\right] \h. 
\end{aligned}
\end{equation*}
The polynomial $p(t)$ can more generally be written as
\begin{equation*}
\begin{aligned}
p(t) &=  \h^T {\boldsymbol \Theta} \h =  \h^T \left[ \sum_{l=0}^{2L-2} t^l {\boldsymbol \Theta}_l  \right]\h 
= (\bar{\h} \otimes \h)^T {\boldsymbol M}^T {\boldsymbol t}
\end{aligned}
\end{equation*}
where the $L \times L$ Hankel matrix ${\boldsymbol \Theta}$ is related to the model order $L-1$, 
\[
{\boldsymbol \Theta}_l  = \left[\begin{array}{ccccc}{\bf 0} &   & 1 &   & {\bf 0} \\  & \iddots &   & \iddots &   \\1 &   & {\bf 0} &   &  \\  & \iddots &   & {\bf 0} &   \\{\bf 0} &   &   &   & \end{array}\right] 
\]
is an $L \times L$ matrix with ones on its $l$th {\it anti-diagonal} and zeros elsewhere (e.g., ${\boldsymbol \Theta}_0$ will have a one on its (1,1) entry and zeros elsewhere),
\[
{\boldsymbol M}^T = [\vec({\boldsymbol \Theta}_0) \cdots \vec({\boldsymbol \Theta}_{2L-2})] \in \mathbb{R}^{L^2 \times 2L-1},
\]
and ${\boldsymbol t}=[1,t,\cdots,t^{2L-2}]$ contains monomials up to order $2(L-1)$. This means that, we can write 
\[
\b(\h) = {\boldsymbol M} ({\bar{\h} \otimes \h}) = {\boldsymbol M} \vec({\h\h^H}),
\]
which together with \eqref{eq:ry_vertex} leads to the constrained least squares:
\[
\underset{{\h_{\rm kr} \in \mathbb{R}^{L^2}}}{\text{minimize}} \quad  \|\r_\y -  {\boldsymbol C}{\boldsymbol h}_{\rm kr}\|_2^2 \quad {\rm s.to} \quad \h_{\rm kr} = {\bar{\h} \otimes \h} 
\]
with ${\boldsymbol C} := ({\boldsymbol \Phi} \otimes {\boldsymbol \Phi}){\boldsymbol \Psi}_{\rm MA}{\boldsymbol M}$. The above least squares problem that accounts for the Kronecker structure in the unknowns can be solved using algebraic methods developed in~\cite{van1996analytical}, or by introducing a rank-1 matrix ${\H}_{\rm kr} = {\h}{\h}^H$ and then solve for $\H_{\rm kr}$ and $\h$ using standard rank relaxation techniques~\cite{luo2010semidefinite}.   
\end{myrem}
}

In sum, if the subsampling matrix $\Phib$ is carefully designed (subject of Section \ref{sec:subsampdesgn}), we can recover the moving average graph power spectrum of a length-$N$ graph signal by observing only $\mathcal{O}(\sqrt{Q})$ nodes.

\subsection{Graph autoregressive models}

A {\it graph autoregressive signal} ({\it G-AR}) of order $P$ may be generated by filtering zero-mean unit-variance white noise, ${\boldsymbol n}$, with an all-pole filter of the form~\cite{marques2016stationary}
\begin{equation}
\label{eq:GARnonlinear}
{\boldsymbol H}{^{-1}}({\boldsymbol \alpha}) = \prod_{k=1}^P ({\boldsymbol I} - \alpha_k {\boldsymbol S}),
\end{equation}
where the {\it G-AR} coefficients $\{\alpha_k\}_{k=1}^P$ are collected in the length-$P$ vector ${\boldsymbol \alpha}$. 
{ Such all-pole filters are useful to model, e.g., diffusion processes~\cite{marques2016stationary} and graph power spectra with sharp transitions.} 

The covariance matrix $\R_\x$ of the {\it G-AR} { signal}, $\x = \H({\boldsymbol \alpha})\n$, given by
\[
\R_\x = \H({\boldsymbol \alpha})\H^H({\boldsymbol \alpha}) \in \mathbb{C}^{{N} \times {N}},
\]
does not admit a linear parameterization in ${\boldsymbol \alpha}$ (unlike the moving average approach that we have seen earlier).
The subsampled covariance matrix $\R_\y \in \mathbb{C}^{{K} \times {K}}$ of the subsampled observations $\y = \Phib\x = \Phib \H({\boldsymbol \alpha})\n \in \mathbb{C}^{K}$, given by 
\[
\R_\y = \Phib \R_\x \Phib^T =    \Phib \H({\boldsymbol \alpha})\H^H({\boldsymbol \alpha})  \Phib^T. 
\]
is also non-linear in ${\boldsymbol \alpha}$. Consequently, vectorizing $\R_\y$ leads to a set of $K^2$ {\it non-linear} equations in $P$ unknowns
\begin{equation}
\label{eq:GARnonlinearvecto}
\r_\y = (\Phib \otimes \Phib) \r_\x = (\Phib \otimes \Phib) \vec({\H({\boldsymbol \alpha})\H^H({\boldsymbol \alpha})}).
\end{equation}
Solving this system of non-linear equations is {\it not trival} {(e.g., it has to be solved using iterative Newton's methods)}. Therefore, in what follows, we will develop a technique for {\it G-AR} modeling as well as for graph sampling so that the G-AR parameters can be recovered using {\it non-iterative} linear estimators.

The all-pole filter \eqref{eq:GARnonlinear} can be alternatively expressed as
\begin{equation}
\label{eq:GARfilter}
{{\H}^{-1}({\boldsymbol a}) =  {\boldsymbol I} - \sum_{k=1}^P a_k {\boldsymbol S}^k},
\end{equation}
%\begin{equation}
%\label{eq:GARfilter}
%{\boldsymbol H}({\boldsymbol S}) =  \frac{1}{{\boldsymbol I} + \sum_{k=1}^p a_p {\boldsymbol S}^p} =  \prod_{k=1}^p ({\boldsymbol I} - b_p {\boldsymbol S})^{-1},
%\end{equation}
where $\{a_k\}_{k=1}^P$ are the so-called {\it G-AR} parameters. Thus, the {\it {\it G-AR}} signal satisfies the equations
\begin{equation}
\label{eq:GARprocess}
{\boldsymbol x} = \sum_{k=1}^P a_k{\boldsymbol S}^k{\boldsymbol x}  + {\boldsymbol n}.
\end{equation}
In other words, the graph signal ${\boldsymbol x}$ depends {\it linearly} on the $P$-shifted graph signals $\{{\boldsymbol S}^k{\boldsymbol x}\}_{k=1}^{P}$ according to the above autoregressive model. So the covariance matrix of $\x$ can be expressed as
\begin{equation}
\label{eq:AR_uncomp}
\R_\x = \sum_{k=1}^P a_k{\boldsymbol S}^k\R_{\boldsymbol x} {+ \R_{\n\x}},
\end{equation}
which is also linear in the G-AR {parameters}, { and where $\R_{\n\x} = \mathbb{E}\{\n\x^H\}$ may be seen as an error term. Given the (uncompressed) observations, $\x$, the above linear model can be used to compute the G-AR coefficients using least squares.}

%%%%%
Let $\mathcal{N}_k(p)$ denote the set of nodes in the $p$-hop neighborhood of the $k$th node, i.e., 
{
\[
%\mathcal{N}_k(p) := \{l \,\, | \,\, l \in \mathcal{N}, [{\boldsymbol A}^{p}]_{k,l {\neq k}} \neq 0\}.
\mathcal{N}_k(p) := \{l \,\, | \,\, l \in \mathcal{N}, [{\boldsymbol S}^{p}]_{k,l} \neq 0\}.
\]
} 
Using this notation, we will now describe the specific subsampling scheme that we adopt for G-AR models, and we will explain later the advantage of this particular subsampling scheme. Suppose we observe $K_0$ graph nodes through a \emph{sparse} subsampling matrix ${\boldsymbol \Phi}_0 \in \{0,1\}^{K_0 \times N}$. Let us denote the set containing the indices of the subsampled nodes by $\mathcal{K}_0$ such that $|\mathcal{K}_0| = K_0$.  
Furthermore, we will also observe nodes in the $P$-hop neighborhood of those $K_0$ nodes through $\{\Phib_p\}_{p=1}^P$. More specifically, with $\Phib_p$ we observe nodes in the set $\mathcal{N}_{k}(p)$ for ${k \in \mathcal{K}_0}$ such that the matrix $\Phib_p$ will have $K_p := \sum_{k \in \mathcal{K}_0} |\mathcal{N}_k(p)|$ rows with ${ \Phib_p \in \{0,1\}^{K_p \times N}}$. Mathematically, the above subsampling scheme ${\y} = \Phib \x$ can be expressed as follows:
\[
{\y} =  [\Phib_0^T, \Phib_1^T, \cdots,\Phib_P^T]^T \x = [\y_0^{T},\y_1^{T},\ldots,\y_P^{T}]^T \in \mathbb{C}^K,
\]
where $\y$ is a vector of length $K = \sum_{l=0}^P K_l$, which is also the total number of observations we gather. This sampling scheme is inspired from~\cite{testa2016compressive}, and we extend it for reconstructing second-order statistics by recognizing the fact that the compressed observations (and their covariance matrices) satisfy the G-AR model.  For the sake of presentation, we make abstraction of the redundancies in the observations $\y$ that may arise due to the nonzero diagonal entries of the powers of the shift-operator or due to overlapping nodes within different neighborhoods. Note that the subsampling scheme for the {\it G-AR} model is different from the subsampling schemes discussed in Sections \ref{sec:subsampling} and \ref{sec:GraphVertexDomain} as we observe a subset of nodes and its related neighborhood as well. For example, suppose each node has degree $n$, then we acquire $\mathcal{O}(K_0[1+ n+ n^2+\cdots+n^P]) = \mathcal{O}(K_0 (1-n^{P+1})/(1-n))$ observations in total.

Using \eqref{eq:GARprocess}, we can express the observations $\y_0 = \Phib_0 \x$ as
\begin{equation}
\label{eq:GAR_sub}
\begin{aligned}
{\boldsymbol y}_0 &= \sum_{k=1}^P a_k {\boldsymbol \Phi}_0 {\boldsymbol S}^k{\boldsymbol x}  + {\boldsymbol \Phi}_0 {\boldsymbol n},\\
&= \sum_{k=1}^P a_k {\boldsymbol \Phi}_0 {\boldsymbol S}^k {\boldsymbol \Phi}_k^T {\boldsymbol y}_k  + {\boldsymbol \Phi}_0 {\boldsymbol n},
\end{aligned}
\end{equation}
where the second equality is due to the structure of the shift operator that operates (locally) on the neighboring nodes, and thus can be expressed via a column selection operation ${\boldsymbol \Phi}_k^T \in \{0,1\}^{N \times K_k}$. { Due to the choice of this particular subsampling scheme, the compressed observation $\y_0$ can be expressed as a linear combination of the compressed observations $\{\y_k\}_{k=1}^P$ with the G-AR parameters being the combining weights.}
%{ 
%due to the fact that $\x$ is a {\it {\it G-AR}} signal
%that $\y$ is also a {\it {\it G-AR}} process {\color{red}(should we change the previous sentence saying that the compression preserves linearity?).}

%%%%
%Suppose we observe only $K$ graph nodes through a \emph{sparse} subsampling matrix 
%${\boldsymbol \Phi} \in \{0,1\}^{K \times N}$ as
%\begin{equation}
%\label{eq:GAR_sub}
%{\boldsymbol y} = {\boldsymbol \Phi} {\boldsymbol x} = \sum_{k=1}^p a_k {\boldsymbol \Phi} {\boldsymbol S}^k{\boldsymbol x}  + {\boldsymbol \Phi} {\boldsymbol n},
%\end{equation}
%where we assume that we have access to the subsampled versions of the $p$-shifted graph signals $\{{\boldsymbol S}^k{\boldsymbol x}\}_{k=1}^{p}$. This implicitly means that, if we observe a node, we also observe it's $p$-hop nodes. More specifically, we subsample the observations
%\[
%\X := [\x, \Sb\x,\cdots,\Sb^p\x] \in \mathbb{C}^{N \times (p+1)},
%\]
%obtained by collecting $\x$ and the first $p$ applications of the shift operator on $\x$ as
%\[
%\Y = \Phib \X \in \mathbb{C}^{K \times (p+1)}.
%\]
%So, in total we gather $pK$ samples.

By defining $\R_{p,q} = \mathbb{E}\{\y_p\y_q^H\} = \Phib_p \R_\x \Phib_q^T \in \mathbb{C}^{K_p \times K_q}$, we can express the covariance matrix $\R_{0,0}$ in terms of the available observations as
\begin{equation}
\label{eq:GAR_subcov}
\begin{aligned}
\R_{0,0} &= \Phib_0 \R_{\boldsymbol x} \Phib_0^T \\%&= \sum_{k=1}^P a_k {\boldsymbol \Phi}_0 {\boldsymbol S}^k \Phib_k^T \Phib_k \R_{\x}\Phib_0^H\\
% \\ &= \sum_{k=1}^p a_p {\boldsymbol \Phi} \R_k \Phib^T. 
&= \sum_{k=1}^P a_k {\boldsymbol \Phi}_0 {\boldsymbol S}^k \Phib_k^T\R_{k,0}  {+  {\boldsymbol \Phi}_0\R_{\n\x} {\boldsymbol \Phi}_0^T},
\end{aligned} 
\end{equation}
%Here, we define $\R_k := {\boldsymbol S}^k\R_{\boldsymbol x}$. 
which on vectorizing leads to $K_0^2$ equations in $P$ unknowns given by
\begin{equation}
\label{eq:GAR_veccov}
\begin{aligned}
\r_{0,0} &= (\Phib_0 \otimes \Phib_0)  \vec(\R_{\boldsymbol x}) \\
&{\approx} \sum_{k=1}^P a_k \vec({\boldsymbol \Phi}_0 {\boldsymbol S}^k {\boldsymbol \Phi}_k^T\R_{k,0}) 
=\, {\boldsymbol G}_{0} \a
\end{aligned}
\end{equation}
where {$\approx$ is due to the error term.} Here, we have stacked $\vec({\boldsymbol \Phi}_0 {\boldsymbol S}^k {\boldsymbol \Phi}_k^T\R_{k,0})$ to form the columns of the matrix ${\boldsymbol G}_{0} \in \mathbb{R}^{K_0^2 \times P}$ { as 
\[
{\boldsymbol G}_{0} = [\vec({\boldsymbol \Phi}_0 {\boldsymbol S} {\boldsymbol \Phi}_1^T\R_{1,0}), \cdots, \vec({\boldsymbol \Phi}_0 {\boldsymbol S}^P {\boldsymbol \Phi}_P^T\R_{P,0})].
\]
}If the $K_0^2 \times P$ matrix ${\boldsymbol G}_{0}$ has full column rank, which requires $K_0^2 \geq P$, then the overdetermined system \eqref{eq:GAR_veccov} can be solved using least squares as
\[
\widehat{\boldsymbol a} := \G_{0}^\dag {\boldsymbol r}_{0,0}.
\] 
Therefore, with a carefully chosen subsampling matrix $\Phib$, we can recover a {\it G-AR} spectrum of a length-$N$ graph signal, residing on a graph with per node degree $n$ with $\mathcal{O}(\sqrt{P}(1-n^{P+1})/(1-n))$ samples.
%for a general case.
% with $\tilde{K} =  \sqrt{P} + \sum_{k=1}^P \sum_{j \in \mathcal{K}_0; |\mathcal{K}_0| = \sqrt{P}} |\mathcal{N}_j(k)|$.
%
%Now the question of interest is, can we further compress (i.e., below $\mathcal{O}(\sqrt{P}(1-n^{P+1})/(1-n))$ observations, for instance) to find the {\it G-AR} parameters $\{a_k\}_{k=1}^P$? To do so, we need more equations that can be obtained as follows. 

Previously in \eqref{eq:GAR_subcov}, we used only the equations related to the covariance matrix of $\y_0$, i.e., $\Phib_0\R_\x\Phib_0^H$, which resulted in $K^2_0$ equations in $P$ unknowns. In addition to this, since we have access to $\{\y_k\}_{k=1}^P$, we can also use the equations corresponding to the covariances between $\y_0$ and observations $\{\y_k\}_{k=1}^P$. This results in the following system
of equations { for $q=0,1,\ldots,P$}:
\begin{equation}
\label{eq:GAR_subblock}
\begin{aligned}
\R_{0,q} &= \Phib_0 \R_\x \Phib_q^T \\
&= \sum_{k=1}^P a_k \Phib_0\Sb^{k} \Phib_k^T \R_{k,q} {+  {\boldsymbol \Phi}_0\R_{\n\x} {\boldsymbol \Phi}_q^T},
\end{aligned}
\end{equation}
where $\R_{0,q} \in \mathbb{C}^{K_0 \times K_q}$.
%where we define $\Q_{k,l} := \R_\x\Sb^{l-1+k}$. While writing the above equations, we use the fact that $\Sb^k \R_\x \Sb^l = \Sb^{k+m} \R_\x \Sb^{k-m}$ for a %graph second-order stationary process.  
Vectorizing $\R_{0,q}$ in \eqref{eq:GAR_subblock} for $q=0,1,\ldots,P$, we get
{
\begin{equation}
\label{eq:GAR_subblockvec}
\begin{aligned}
\r_{0,q} &=(\Phib_q \otimes \Phib_0)  \vec(\R_{\boldsymbol x})  \\
&\approx   \sum_{k=1}^P a_k \vec\left(\Phib_0\Sb^{k} \Phib_k^T \R_{k,q}\right) = \G_{q} \a,
\end{aligned}
\end{equation}
}where we have stacked $\vec\left(\Phib_0\Sb^{k} \Phib_k^T \R_{k,q}\right)$ to form the columns of the matrix ${\boldsymbol G}_{q} \in \mathbb{R}^{K_0K_q \times P}$ as 
\[
\G_{q} = \left[\vec(\Phib_0\Sb \Phib_1^T \R_{1,q}), \cdots, \vec(\Phib_0\Sb^{P} \Phib_P^T \R_{P,q})\right].
\]
Now, collecting $\{\r_{0,q}\}_{q=0}^{P}$ in $\r_y$ as 
\[
\r_\y = [\r_{0,0}^T,\r_{0,1}^T, \ldots, \r_{0,P}^T]^T,
\] 
and 
$\{\G_{q}\}_{q=0}^{P}$ in $\G$ as 
\[
\G = [\G_{0}^T,\G_{1}^T,\cdots,\G_{P}^T]^T,
\] 
we have $K_0\sum_{q=0}^PK_q$ equations in $P$ unknowns, i.e.,
{
\begin{equation}
\label{eq:GAR-ry}
\begin{aligned}
\r_\y = (\Phib \otimes \Phib_0) \vec(\R_\x)&=  \G \a. 
\end{aligned}
\end{equation}
where recall that $K = K_0 \sum_{q=0}^PK_q$.} It can be shown that the observation matrix $\G$ can be expressed as {$(\Phib \otimes \Phib_0) \Psib_{\rm AR}$} for some matrix $\Psib_{\rm AR}$ (``AR" stands for {\it autoregressive}), which now depends on the compressed observations, sampling matrices, and the graph shift operator.

The above linear system \eqref{eq:GAR-ry} can be solved using least squares as
\[
\widehat{\a} = \G^\dag {\r}_{\y}
\]
if the observation matrix $\G$ has full column rank. This requires $K_0\sum_{q=0}^PK_q \geq P$. Suppose the graph is connected such that every node has at least one neighbor, then by picking one node would already lead to an overdetermined system. In other words, we can recover a {\it G-AR} spectrum with $K_0 = 1$, which amounts to observing { more than $P$ nodes}. For example, recall the cycle graph in Fig.~\ref{fig:cyclegraph} with $N$ nodes, where every node has a degree of two. In order to recover two {\it G-AR} parameters on such graphs (more generally, for any arbitrary graph with per node degree 2) we need to observe {at least $K_0 + K_1 + K_2 =5$} nodes using this technique. 
{ Depending on the graph, this scheme as such might not lead to any compression at all (e.g., in dense graphs) because all $N$ nodes might be in these $K_0P$-hop neighborhoods. In other words, the proposed scheme is more useful for sparse graphs or with small $P$.}

\section{Finite Data Records} \label{sec:finitedata}
So far to recover the graph second-order statistics we have assumed that the true compressed covariance matrix $\R_\y = \mathbb{E}\{\y\y^H\} \in \mathbb{C}^{K \times K}$ is available. However, in practice we only have a finite number of snapshots, call it ${N_s}$, available.  Suppose we observe ${N_s}$ subsampled graph signals denoted by the vectors $\{\y[k]\}_{k=1}^{{Ns}}$, and they are collected in a $K \times {N_s}$ matrix $\Y := \left[\y[1],\y[2],\ldots,\y[{N_s}]\right]$. It is common to use the sample data covariance matrix $\widehat{\R}_\y =\frac{1}{{N_s}} \Y\Y^H \in \mathbb{C}^{K \times K}$ as an estimate of $\R_\y$. We have seen in Sections \ref{sec:subsampling} and~\ref{sec:parametric} that the compressed covariance matrix $\R_\y$ has a special (linear) structure and it is parameterized by a small number of parameters $\thetab$. In this section, we will provide the least squares estimator, maximum likelihood estimator, and the Cram\'er-Rao lower bound for this finite data records scenario.

Let us denote the structured matrix $\R_\y$ as $\R_\y(\thetab)$. Generally, $\r_\y = \vec(\R_\y (\thetab))$  can be expressed as 
\begin{equation}
\label{eq:cov_structure}
\r_\y := \G \thetab,
\end{equation}
where from \eqref{eq:subsamCovVec} we have $\G := (\Phib \otimes \Phib)\Psib_{\rm s}$ and $\thetab := \p$ for the nonparametric spectral domain approach, from \eqref{eq:ry_vertex} we have $\G := (\Phib \otimes \Phib)\Psib_{\rm MA}$ and $\thetab := \b$ for the parametric moving average model, and
from \eqref{eq:GAR-ry} we have {$\G := (\Phib \otimes \Phib_0)\Psib_{\rm AR}$} and $\thetab := \a$ for the parametric autoregressive model. Before we present the least squares solution in the next subsection, we recall that, although we perform a linear compression on $\R_\x$ as $\R_\y = \Phib \R_\x \Phib^T$, the linear structure in $\R_\x(\thetab)$ is maintained in $\R_\y(\thetab)$ as well, as long as the compression matrix is a valid covariance subsampler.

\subsection{Least squares estimator}
Under the abstraction in \eqref{eq:cov_structure}, the question now is, how can the estimated covariance matrix $\widehat{\r}_\y = \vec(\widehat{\R}_\y)$ be {\it matched} to the true covariance matrix $\R_\y$, which has a linear structure. This can for instance be solved in the least squares sense as
%
%The problem of fitting the estimated covariance matrix $\widehat{\R}_\y$ to the structured covariance matrix $\R_\y$ is then simply solving the optimization
%problem 
%, e.g., in the least squares sense
%This can be solved by
\begin{equation}
\label{eq:ls_finite}
\widehat{\thetab} = \argmin_{\thetab} \|\widehat{\r}_\y - \G\thetab\|_2^2 = \G^\dag\widehat{\r}_\y.
\end{equation}
Therefore, to summarize, the results derived so far in this paper (including estimators and subsampler designs) for infinite data records are also valid for scenarios with finite data records. Furthermore, the above least squares problem may be also solved with a constraint on $\thetab$, which leads to a constrained least squares problem [cf. Remarks~\ref{rem:cls_prior} and~\ref{rem:cls_ma}]. 

{ The least squares estimators derived thus far do not assume any data distribution and they are reasonable for any data probability density function. In what follows, we will discuss a special case, where the observations are Gaussian distributed.}

\subsection{Maximum likelihood estimator and Cram\'er-Rao bound}

Suppose the compressed data consists of realizations from a sequence of independent and identically distributed (i.i.d.) Gaussian random vectors $\{\y[k]\}_{k=1}^{N_s}$, where for each $k$, the length-$K$ vector $\y[k] \thicksim \mathcal{CN}({\boldsymbol 0}, \R_\y(\thetab))$ with the (positive definite) covariance matrix $\R_\y(\thetab)$ being a function of the parameters $\thetab$ as in \eqref{eq:cov_structure}.

The maximum likelihood estimate of $\thetab$ given ${\boldsymbol Y}$ is obtained by solving the optimization problem
\[
\widehat{\thetab} = \argmax_{\thetab} \,\, l({\boldsymbol Y};\thetab)
\]
with log-likelihood function { (with terms that depend only on the unknowns)~\cite{ottersten1998covariance,scharf1991statistical}}
%\[
%\small
%p({\boldsymbol Y};\thetab) = (2\pi)^{-{N_s}K/2} (\det {\R_\y})^{-{N_s}/2} \exp \left\{-\frac{{N_s}}{2} {\rm tr}\{{\R_\y^{-1}}\widehat{\R}_\y\}\right\}.
%\]
{
\[
l(\Y;\thetab) = {\nu N_s}\left[ \log \det \{\R_\y^{-1}(\thetab)\} - {\rm tr}\{{\R_\y^{-1}(\thetab)}\widehat{\R}_\y\} \right],
\]
where $\nu=1$ if $\R_\y$ has complex entries and $\nu=0.5$ if $\R_\y$ has real entries.}

The maximum likelihood estimate of $\thetab$ can then be computed by setting the derivative of $l(\Y;\thetab)$ with respect to $\thetab$ to zero, and it is the solution to the regression equation~\cite{scharf1991statistical}:
\begin{equation}
\label{eq:ML}
%{\boldsymbol g}_i^H (\R_\y^{-T} \otimes \R_\y^{-1}) \vec(\R_\y - \widehat{\R}_\y) = 0, \quad \forall i,
{\boldsymbol g}_i^H [\R_\y^{-T} \otimes \R_\y^{-1}] (\r_\y - \widehat{\r}_\y) = 0, \quad \forall i,
\end{equation}
where ${\boldsymbol g}_i$ is the $i$th column of $\G$. The above equations must be solved iteratively using algorithms provided in~\cite{burg1982estimation,romero2013wideband,ottersten1998covariance,stoica2011maximum}. { The above equations would hold, if $\r_\y = \widehat{\r}_\y$. The solution \eqref{eq:ls_finite} approximates $\r_\y \approx \widehat{\r}_\y$, in the least squares sense. Also, from \eqref{eq:ML}, we can recognize that the maximum likelihood estimator reduces to a weighted least squares problem
\begin{equation*}
%\label{eq:WLS_finite}
%\begin{aligned}
\argmin_{\thetab}(\widehat{\r}_\y - \G\thetab)^H {\boldsymbol C}_w (\widehat{\r}_\y - \G\thetab) = (\G^H  {\boldsymbol C}_w \G)^{-1}\G^H{\boldsymbol C}_w\widehat{\r}_\y
%\end{aligned}
\end{equation*}
with weighting matrix ${\boldsymbol C}_w = {\nu N_s} (\R_\y^{-T}(\thetab) \otimes \R_\y^{-1}(\thetab))$. For the weighting matrix, we may use the estimate $\widehat{\boldsymbol C}_w$ obtained by using $\widehat{\R}_\y$ instead of $\R_\y$.}

Next, we will provide the Cram\'er-Rao bound, which is a lower bound on the variance of the developed least squares estimators when the available data records are finite. { (Note that this is a bound on the variance of $\widehat{\p}$ obtained from the nonparametric approach, and the Cram\'er-Rao bound for the power spectrum estimates from the parametric methods may be derived using transformation of parameters.)} The Cram\'er-Rao bound matrix is the inverse of the Fisher information matrix. The $(i,j)$th entry of the Fisher information matrix, ${\boldsymbol F}$, is given by~\cite{scharf1991statistical}
\begin{equation}
\begin{aligned}
\label{eq:crb}
[{\boldsymbol F}]_{i,j} &= - \mathbb{E}\left\{ \frac{\partial^2}{\partial [\thetab]_i \partial [\thetab]_j} l(\Y;\thetab) \right\}  \\
&= \nu{N_s}{\boldsymbol g}_j^H [\R_\y^{-T}(\thetab) \otimes \R_\y^{-1}(\thetab)] {\boldsymbol g}_i.
\end{aligned}
\end{equation}
{It can be seen from the expression of the Cram\'er-Rao bound that the developed least squares estimators ignore the color of the residual, $\widehat{\r}_\y - {\r}_\y$, which has a covariance matrix ${\boldsymbol C}_w^{-1}$ (not scaled identity). This means that the developed estimators are not {\it efficient} (i.e., they will not achieve the Cram\'er-Rao bound), but are computationally cheap as compared to the {\it asymptotically efficient} maximum likelihood estimators.}

%We now focus our attention on the design of the compression matrix that is used to sample the second-order statistics.
%% $\P$

%For the graph in Fig.~\ref{fig:cyclegraph}, this means that, to recover $P$ {\it G-AR} parameters we need to observe $2P+2$ nodes, e.g., we observe nodes with labels $\{x_1,x_2, \cdots, x_{P+1}, x_{N-P+1}, \cdots, x_{N-1},x_N\}$.

%We end this section by stressing the achieved significant gain with the following example. Consider a graph with $N=100$ nodes, then sampling just around $\sqrt{N} =10$ nodes, i.e., with about $90\%$ compression, entire graph power spectrum can be estimated. 
%
%A stronger compression can be achieved by using sparsity priors on the graph spectra, if any.  

\section{Sparse Sampler Design} \label{sec:subsampdesgn}

We have seen so far that the design of the subsampling matrix $\Phib$ is crucial for the reconstruction of the graph second-order statistics. From Theorem~\ref{theo:fullrank_spectral}, we know the conditions under which a subsampling matrix will be a valid covariance subsampler, but still it has to be designed. Alternatively, random compression matrices drawn from a certain probability space (e.g., entries of the subsampling matrix are drawn from a Gaussian or Bernoulli distribution) may be used as they almost surely satisfy the conditions in Theorem~\ref{theo:fullrank_spectral} {(see e.g.,~\cite{romero2015compression})}. However, they might not be practical in the graph setting, because random compression matrices are usually {\it dense} in nature, and to compute linear combinations of  the uncompressed graph signals they have to be made available at a central location. On the other hand, if we choose a sparse sampling matrix, which essentially does node selection, only the subsampled graph signals (very few samples as compared to the number of nodes) have to be processed. Therefore, in what follows, we will develop an algorithm to design a sparse subsampling matrix. 

Consider a structured sparse sampling matrix ${\boldsymbol \Phi} \in \{0,1\}^{K \times N}$, such that the entries of this matrix are determined 
by a binary sampling vector $\w$. More specifically, let us denote the structured subsampling matrix ${\boldsymbol \Phi}$ as 
${\boldsymbol \Phi}({\boldsymbol w}) = {\rm diag_r}[{\boldsymbol w}] \in \{0,1\}^{K \times N}$, which is guided by a {\it component selection} vector  ${\boldsymbol w} = [w_1,\cdots,w_N]^T\in \{0,1\}^N$, where $w_i=1$ indicates that the $i$th graph node is selected, otherwise it is not selected. That is, ${\boldsymbol \Phi}({\boldsymbol w})$
essentially performs {\it graph sampling.}
% (${\rm diag_r}[\cdot]$ represents a diagonal matrix with the argument on its diagonal but with the all-zero rows removed)

\subsection{Spectral domain and moving average case} \label{sec:grdy_manonpar}
In this subsection, we will design the subsampling matrix for the estimators based on the spectral domain approach [cf. Section~\ref{sec:subsampling}] and the vertex domain parametric moving average model [cf. Section~\ref{sec:GraphVertexDomain}] as the observation matrices in these cases share a common structure. In particular, the aim is to design a full-column rank observation matrix 
${\G} = [{\boldsymbol \Phi}(\w) \otimes {\boldsymbol \Phi}(\w)]{\boldsymbol \Psi}$ with ${\boldsymbol \Psi}:={\boldsymbol \Psi}_{\rm s}$ or  ${\boldsymbol \Psi}:={\boldsymbol \Psi}_{\rm MA}$, so that we can perfectly recover the second-order statistics by observing a reduced set of only $K$ graph nodes. {To do this, we assume $\Psib$ is perfectly known.}

{Uniqueness and sensitivity} of the least squares solution developed in Sections \ref{sec:subsampling} and \ref{sec:GraphVertexDomain} depends on the spectrum {(i.e., the set of eigenvalues)} of the matrix  
\begin{align*}
{\boldsymbol T}({\boldsymbol w}) &= [({\boldsymbol \Phi}({\boldsymbol w}) \otimes {\boldsymbol \Phi}({\boldsymbol w})){\boldsymbol \Psi}]^T[({\boldsymbol \Phi}({\boldsymbol w}) \otimes {\boldsymbol \Phi}({\boldsymbol w})){\boldsymbol \Psi}]\\
&{=} {\boldsymbol \Psi}^T ({\rm diag}[{\boldsymbol w}] \otimes {\rm diag}[{\boldsymbol w}]) {\boldsymbol \Psi}.
\end{align*}
In other words, the performance of least squares is better if the spectrum of the matrix $({\boldsymbol \Phi} \otimes {\boldsymbol \Phi}){\boldsymbol \Psi}$ is more uniform~\cite{golub1996matrix}. Thus, a good sparse sampler ${\boldsymbol w}$ can be obtained by solving:
%\begin{equation}
%\label{eq:eigopt}
%\argmax_{{\boldsymbol w} \in \{0,1\}^N} \quad f({\boldsymbol T}({\boldsymbol w}) \quad{\rm s.t.}\quad  \|{\boldsymbol w}\|_0 = K
%\end{equation}
%or
\begin{equation}
\label{eq:logdetopt}
\argmax_{{\boldsymbol w} \in \{0,1\}^N} \quad f({\boldsymbol w}) \quad{\rm s.to}\quad  \|{\boldsymbol w}\|_0 = K
\end{equation}
with either $f({\boldsymbol w}) = - {\rm tr}\{{\boldsymbol T}^{-1}({\boldsymbol w})\}$, $f({\boldsymbol w}) = \lambda_{\rm min}\{{\boldsymbol T}({\boldsymbol w})\}$, or $f({\boldsymbol w}) = \log\det \{{\boldsymbol T}({\boldsymbol w})\}$, which tries to balance the spectrum of ${\boldsymbol T}({\boldsymbol w})$.  
Alternatively, the Fisher information matrix \eqref{eq:crb} can be used instead of ${\boldsymbol T}({\boldsymbol w})$ to design samplers using techniques discussed in~\cite{ChepuriTSPsel}. 

\subsubsection{Convex relaxation}
The above Boolean nonconvex problem with any one of the cost functions can be relaxed and solved using convex optimization (e.g., see~\cite{ChepuriTSPsel, chepuri2014spl}). To express \eqref{eq:logdetopt} as a convex optimization problem, we will introduce an auxiliary variable ${\boldsymbol Z} = {\w}{\w}^T$ and its related length-$N^2$ vector ${\boldsymbol z} := \vec({\boldsymbol Z})$.  
{ Since ${\rm diag}[{\boldsymbol w}] \otimes {\rm diag}[{\boldsymbol w}] = {\rm diag}[{\boldsymbol z}]$, we can write $f({\boldsymbol w})$ as $f({\boldsymbol z})$}, and relaxing (a) Boolean constraints on $\w$ to the box constraints, (b) the cardinality constraint to an $\ell_1$-norm constraint, and (c) the rank-1 constraint on ${\boldsymbol Z}$, we obtain the following optimization problem
\begin{equation}
\begin{aligned}
\label{eq:cvx_case1}
&\argmax_{{\boldsymbol w}, {\boldsymbol Z}} \quad f({\boldsymbol z}) \\
 &\quad {\rm s.to}\quad  {\bf 1}^T\w = K, \quad 0\leq w_n\leq 1, n=1,\ldots,N,\\
 &\,\, \quad \quad \quad {\boldsymbol Z} \succeq \w\w^T, \,\, {\boldsymbol z} = \vec({\boldsymbol Z}),
 \end{aligned}
\end{equation}
where ${\boldsymbol Z} \succeq \w\w^T$ can be expressed as a linear matrix inequality that is linear in $\w$. 

\subsubsection{Submodular greedy optimization}
Due to the involved complexity of solving the convex relaxed problem \eqref{eq:cvx_case1} and keeping in mind the large scale problems that arise in the graph setting, we will now focus on the optimization problem \eqref{eq:logdetopt} with $f({\boldsymbol w}) = \log\det \{{\boldsymbol T}({\boldsymbol w})\}$ as it can be solved near-optimally using a low-complexity greedy algorithm.

Let us define an index set $\mathcal{X}$ that is related to the component selection vector ${\boldsymbol w}$ as 
$
\mathcal{X} = \{ m \, | \, w_m=1,m=1,\ldots,N\},
$ 
where $\mathcal{X} \subseteq \mathcal{N}$ with $\mathcal{N} = \{1,\ldots,N\}$.
We can now express the cost function $f({\boldsymbol w}) = \log\det \{{\boldsymbol T}({\boldsymbol w})\}$ equivalently as the set function 
given by 
\begin{equation}
\label{eq:submodularlogdet}
f(\mathcal{X}) = \log \det \left\{\sum\limits_{{(i,j)} \in \mathcal{X} \times \mathcal{X}} {\boldsymbol \psi}_{i,j} {\boldsymbol \psi}_{i,j}^T\right\},
\end{equation}
where the length-$N^2$ column vectors $\{{\boldsymbol \psi}_{1,1},{\boldsymbol \psi}_{1,2}, \cdots,{\boldsymbol \psi}_{N,N}\}$ are used to form the rows of ${\boldsymbol \Psi}$ as ${\boldsymbol \Psi} = [{\boldsymbol \psi}_{1,1},{\boldsymbol \psi}_{1,2}, \cdots,{\boldsymbol \psi}_{N,N}]^T$. We use such an indexing because the sampling matrix ${\boldsymbol \Phi} \otimes {\boldsymbol \Phi}$ results in a structured (row) subset selection. The notation $\sum\nolimits_{(i,j)}$ denotes the double summation; As an example, for $\mathcal{X} = \{1,2\}$, we have $\sum\nolimits_{(i,j) \in \mathcal{X} \times \mathcal{X}} {\boldsymbol \psi}_{i,j} = {\boldsymbol \psi}_{1,1} + {\boldsymbol \psi}_{1,2} + {\boldsymbol \psi}_{2,1}+ {\boldsymbol \psi}_{2,2}$. 

{\it Submodularity} \textemdash a notion based on the property of diminishing returns, is useful for solving discrete combinatorial optimization problems of the form \eqref{eq:logdetopt} (see e.g., \cite{krause2008optimizing}).  Submodularity can be formally defined as follows.

\begin{mydef}[Submodular function]
Given two sets $\mathcal{X}$ and $\mathcal{Y}$ such that for every $\mathcal{X} \subseteq \mathcal{Y} \subseteq \mathcal{N}$ and $s \in \mathcal{N} \backslash \mathcal{Y}$, the set function $f: 2^{N} \rightarrow \mathbb{R}$ defined on the subsets of $\mathcal{N}$ is said to be submodular, if it satisfies
\[
f(\mathcal{X} \cup \{s\}) - f(\mathcal{X}) \geq f(\mathcal{Y} \cup \{s\}) - f(\mathcal{Y}).
\]
\end{mydef}

Suppose the submodular function is monotone nondecreasing, i.e., $f(\mathcal{X})$ $\leq f(\mathcal{Y})$ for all $\mathcal{X} \subseteq \mathcal{Y} \subseteq \mathcal{N}$ and normalized, i.e., $f(\emptyset) = 0$, then a greedy maximization of such a function as summarized in Algorithm~\ref{alg:greedy} is \emph{near optimal} with an approximation factor of $(1 -1/e)$, where $e$ is Euler's number~\cite{nemhauser1978analysis}. In other words, we can achieve 
\[
f(\mathcal{X}) \geq (1 -1/e) f({\rm opt}),
\]
where  $f({\rm opt})$ is the optimal value of the problem 
\[
\underset{{\mathcal{X} \subseteq \mathcal{N}, |\mathcal{X}| = K}}{\text{maximize}} \,\, f(\mathcal{X}).
\]

In order to have a non-empty input set $f(\emptyset) = 0$, the cost function \eqref{eq:submodularlogdet} is slightly modified with a diagonal loading, and it satisfies the above properties as stated in the following theorem.
\begin{mytheo} \label{lem:submodular_logdet}
The set function $f: 2^{N} \rightarrow \mathbb{R}$ given by
\begin{align}
\label{eq:submodularlogdet_normalized}
\hskip-2mmf(\mathcal{X}) = \log \det \left\{\sum\limits_{{(i,j)} \in \mathcal{X} \times \mathcal{X}} {\boldsymbol \psi}_{i,j} {\boldsymbol \psi}_{i,j}^T + \epsilon {\boldsymbol I}\right\} - N \log \epsilon
\end{align}
is a normalized, nonnegative monotone, submodular function on the set $\mathcal{X} \subset \mathcal{N}$. Here, $\epsilon >0$ is a small constant. 
%Hence \eqref{eq:submodularlogdet_normalized} is a reasonable approximation of \eqref{eq:submodularlogdet}.  
\end{mytheo}
%
%\begin{proof}
%
%
%\end{proof}
In \eqref{eq:submodularlogdet_normalized}, $\epsilon{\boldsymbol I}$ is needed to carry out the first few iterations of Algorithm~\ref{alg:greedy} and $- N \log \epsilon$ ensures that $f(\emptyset)$ is zero.  Using the result from~\cite{shamaiah2010greedy} that the set function $g: 2^N \rightarrow \mathbb{R}$, given by 
\begin{align}
\label{eq:classicallogdet}
g(\mathcal{X}) = \log \det \left\{\sum\limits_{i \in \mathcal{X}} {\boldsymbol a}_{i} {\boldsymbol a}_{i}^T + \epsilon {\boldsymbol I}\right\} 
- N \log \epsilon
\end{align}
with column vectors $\{{\boldsymbol a}_{i}\}_{i=1}^N$ is a normalized, nonnegative monotone, submodular function on the set $\mathcal{X} \subseteq \mathcal{N}$, we can prove Theorem~\ref{lem:submodular_logdet}. Therefore, the solution based on the greedy algorithm summarized in Algorithm~\ref{alg:greedy} results in a $(1-1/e)$ optimal solution for~\eqref{eq:logdetopt}. Note that the number of summands 
in \eqref{eq:classicallogdet} and \eqref{eq:submodularlogdet_normalized}, is respectively, $|\mathcal{X}|$ and $|\mathcal{X}|^2$.
It is worth mentioning that the greedy algorithm is linear in $K$, while computing \eqref{eq:submodularlogdet_normalized} remains the dominating cost. %Nevertheless, \eqref{eq:submodularlogdet_normalized} can be computed efficiently using rank-1 updates, similar to \cite{shamaiah2010greedy}.  

Other submodular functions that promote full-column rank model matrices, e.g., the frame potential~\cite{vertterlisensorplacement} defined as $f({\boldsymbol w}) = {\rm tr}\{{\boldsymbol T}^H({\boldsymbol w}){\boldsymbol T}({\boldsymbol w})\}$, are also reasonable costs to optimize. 
Finally, random subsampling (i.e., $\w$ has random $0$ or $1$ entries) is not suitable as it might not always result in a full-column rank model matrix.

\begin{algorithm}[!t]
\caption{Greedy algorithm}
\label{alg:greedy}
\begin{algorithmic}
\item[1.] \textbf{Require} $\mathcal{X} = \emptyset, K$.
\item[2.] \textbf{for} $k=1$ to $K$
\item[3.] \hskip1cm$
s^\ast = \argmax\limits_{ s \notin \mathcal{X}}   \, f (\mathcal{X} \cup \{s\}) %- f (\mathcal{X})
$
\item[4.] \hskip1cm$
\mathcal{X} \leftarrow \mathcal{X} \cup \{s^\ast\}
$
\item[5.] \textbf{end} 
\item[6.] \textbf{Return} $\mathcal{X}$
\end{algorithmic}
\end{algorithm} 

\subsection{Autoregressive case}

The subsampling matrix for the spectral domain and moving average approaches can be designed offline as the observation matrix $\Psib$ was not depending on the data, but it depends only on the graphical model {(i.e, either $\U$ or $\Sb$)}. In contrast, an {\it optimal} offline subsampler design for the autoregressive case is {\it not possible} due to the fact that the observation matrix depends on the data, and to choose the best subset of nodes requires observations from all the nodes. This is the side effect of modeling the graph autoregressive {signal} as \eqref{eq:GARfilter} to arrive at an elegant linear estimator.

Nevertheless, suppose the second-order statistics are available, e.g., from training data, estimated from subsampled observations using the nonparametric or moving average approach (where the sampler is designed using Algorithm~\ref{alg:greedy} as discussed in Section~\ref{sec:grdy_manonpar}), or by approximating the second-order statistics with white noise, then a suboptimal sampler can be designed with techniques similar as those in Section~\ref{sec:grdy_manonpar}.  

Alternatively, if a high-complexity non-linear estimator can be afforded, then by modeling the graph autoregressive process using~\eqref{eq:GARnonlinear}, 
the dependence of the observation matrix on the data can be avoided [cf.~\eqref{eq:GARnonlinearvecto}]. In that case, the subsampler can be designed offline using techniques in~\cite{ChepuriTSPsel,Shilpa15camsap1}. %However, solving the non-linear equations~\eqref{eq:GARnonlinearvecto} for ${\boldsymbol \alpha}$ remains complicated. 

We underline that the algorithms provided here to design sparse samplers for different cases can also be used to design mean squared error optimal sparse samplers for the compressive covariance sensing framework~\cite{ariananda2012compressive,romero2013wideband,Romero16CCSspm}. In other words, although minimal sparse rulers satisfy the identifiability conditions to reconstruct the second-order statistics of stationary time-series, the algorithms provided in this paper are needed to guarantee a desired reconstruction performance.
\begin{figure*}[!th]
     \centering
\begin{subfigure}[b]{0.45\columnwidth}
     \centering
\includegraphics[width=\columnwidth,height=1.25in]{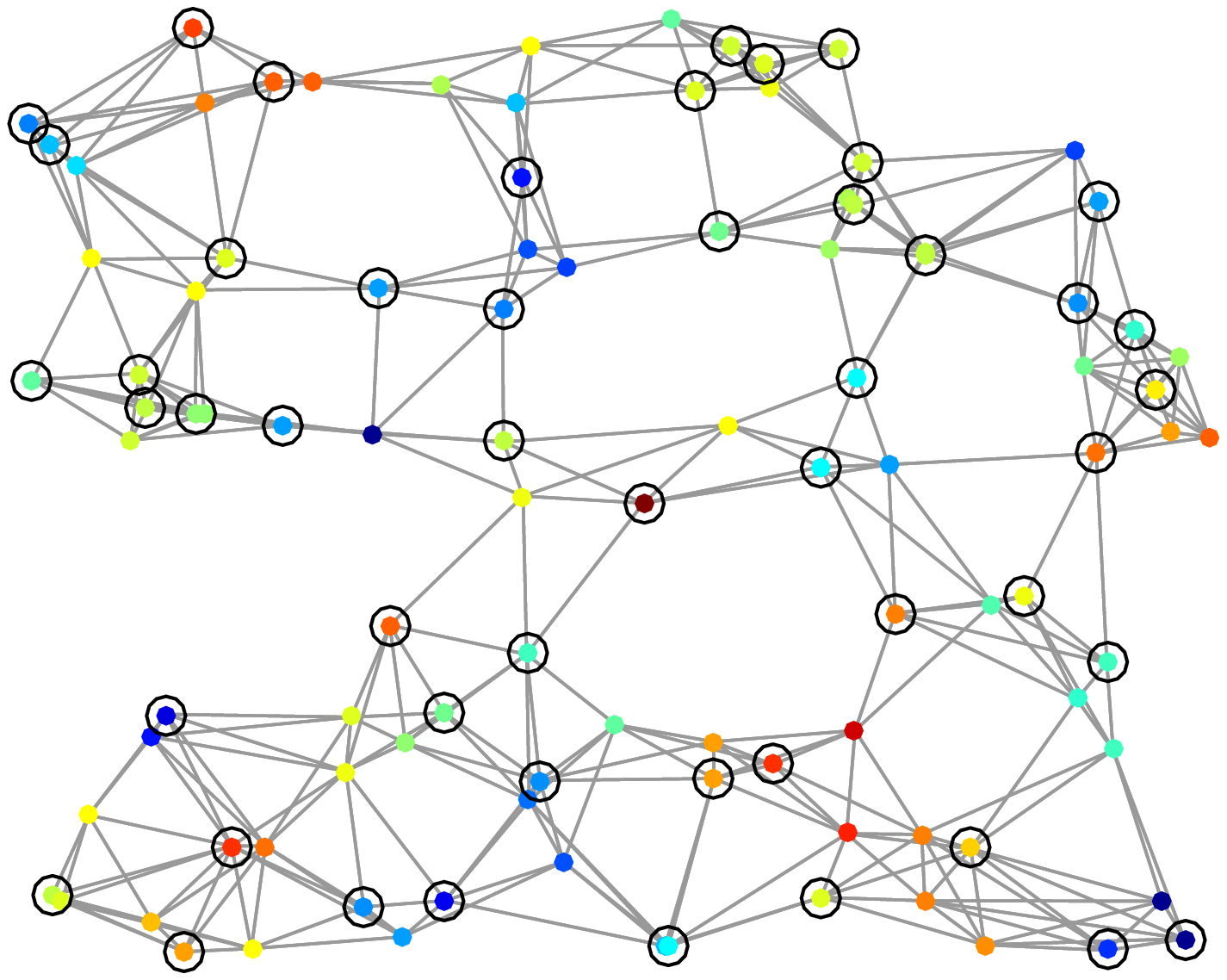}
\caption{}
\label{fig:nodesample1_nonpar}
        \end{subfigure}
~\hskip1.65cm
\begin{subfigure}[b]{0.45\columnwidth}
     \centering
     \includegraphics[width=\columnwidth,height=1.25in]{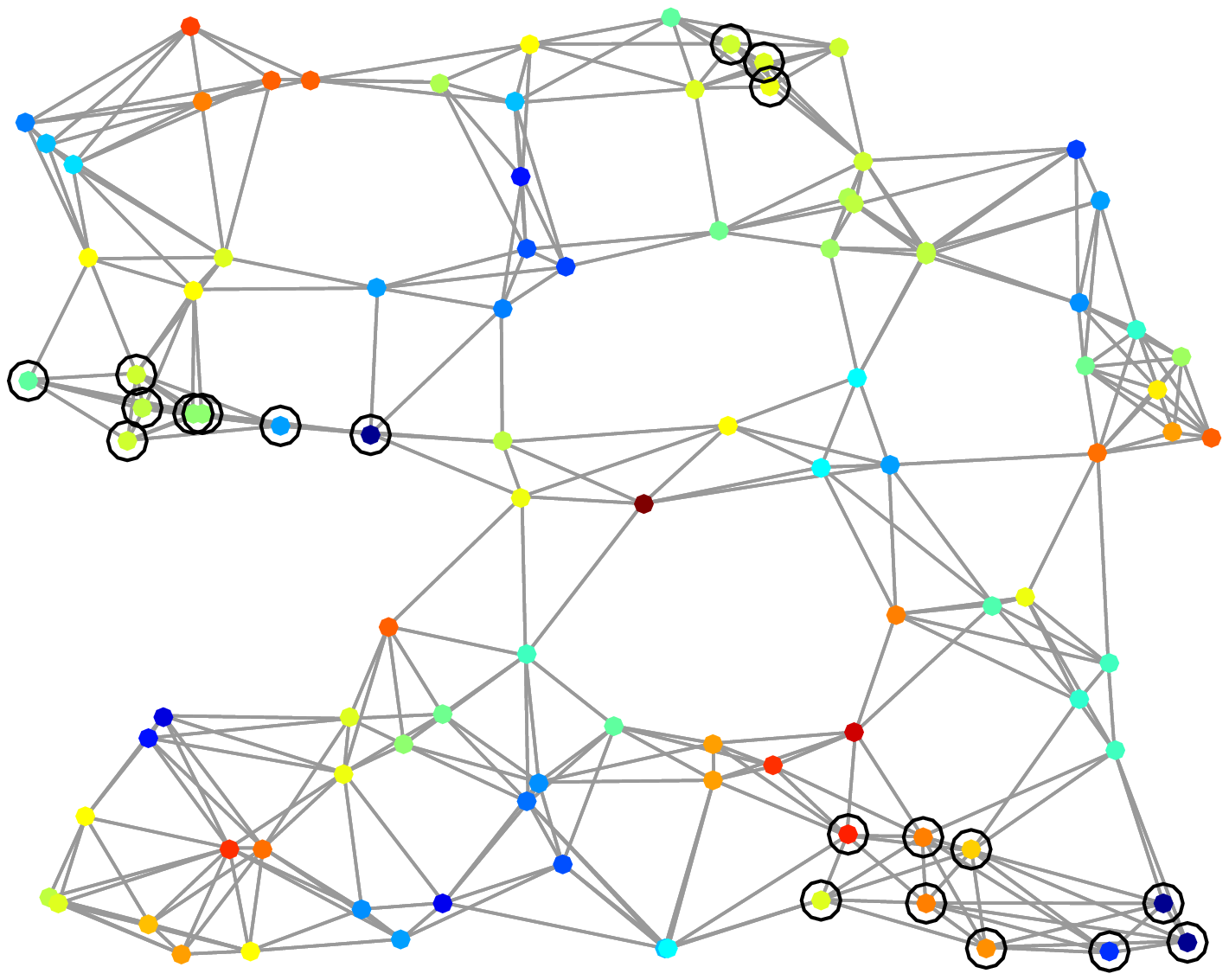}
     \caption{}
     \label{fig:nodesample1_ma}
             \end{subfigure}
~\hskip1.65cm
\begin{subfigure}[b]{0.45\columnwidth}
     \centering
\includegraphics[width=\columnwidth,height=1.25in]{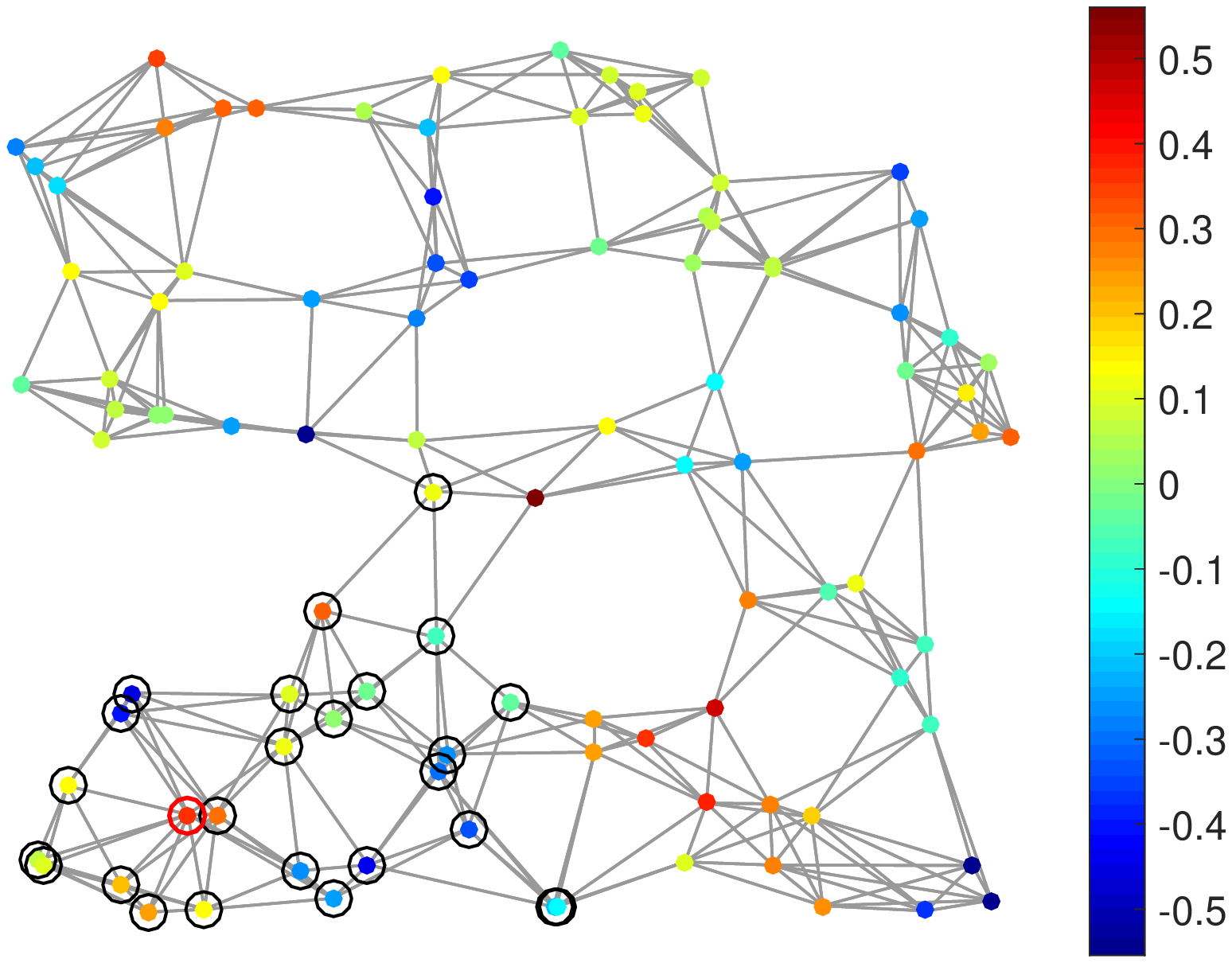}
\caption{}
\label{fig:nodesample1_ar}
             \end{subfigure}
     \caption{Sampling random graphs with $N=100$ nodes for synthetic data. The sampled graph nodes are highlighted by the circles around the nodes and the node coloring simply denotes a realization of the graph signal. (a) Non-parametric model with $K = 50$. (b) Moving average model with $K = 26$. (c) Autoregressive model with $K_0=1$, where the $P$-hop neighborhood around the node indicated with the red circle is observed.}
     \label{fig:nodesamp_synthetic}
\end{figure*}
\begin{figure*}
\psfrag{Laplacian eigenvalues}{\hskip-7mm\footnotesize Shift-operator (Laplacian) eigenvalues}
\psfrag{True PSD}{ \scriptsize True graph power spectrum}
\psfrag{normalized squared error [dB]}{\hskip5mm \footnotesize $\| \p - \widehat{\p}\|_2^2/\|\p\|$}
\psfrag{Percentage Compression}{\hskip5mm\footnotesize$\%$ compression}
\psfrag{Estimated PSD autoreggresive model 76 compression}{\scriptsize Autoregressive ($P=3$, $74\%$ compression) }
\psfrag{Estimated PSD moving average model 80 compression}{\scriptsize Moving average ($Q=13$, $74\%$ compression)}
\psfrag{non-parametric approach 50 compression}{ \scriptsize $50\%$ compression}
\psfrag{non-parametric approach 50+ compression}{ \scriptsize $50\%$ compression}
\psfrag{non-parametric approach 0 compression}{ \scriptsize $0\%$ compression}
\psfrag{non-parametric approach}{\scriptsize Non-parametric}
\psfrag{moving average}{\scriptsize Moving average ($Q=13$)}
\psfrag{CRLB 50 compression}{ \scriptsize CRLB  ($50\%$ compression)}
    \psfrag{normalized mean squared error [dB]}{\hskip10mm\footnotesize NMSE [dB]}
    \psfrag{Number of samples}{\hskip-6mm \footnotesize Number of snapshots, $N_s$}
    \psfrag{moving average 74 compression}{\scriptsize $74$\% compression}
        \psfrag{moving average 90 compression}{\scriptsize $90$\% compression}
    \psfrag{moving average 0 compression}{\scriptsize $0$\% compression}
        \psfrag{autoreggresive 0 compression}{\scriptsize $P=3$, $0$\% compression}
        \psfrag{autoreggresive 74 compression}{\scriptsize $P=3$, $74$\% compression}
        \psfrag{autoreggresive 0 compression P6}{\scriptsize $P=6$, $0$\% compression}
     \centering
          \begin{subfigure}[b]{0.4\textwidth}
               \centering
          {\includegraphics[width=2.9in, height=1.65in]{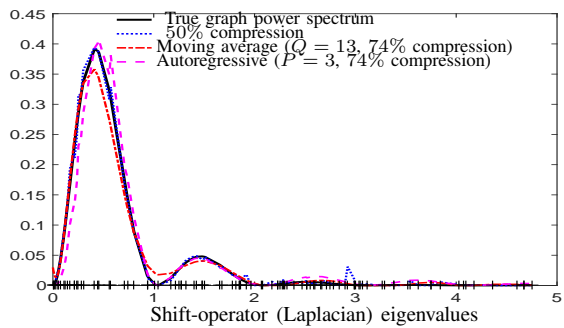}}
          \caption{Graph power spectrum based on $N_s=1000$ snapshots}
          \label{fig:psd_syn}
          \end{subfigure} 
 ~\hskip1.25cm
%     \begin{subfigure}[b]{0.48\textwidth}
%          \centering
%     \includegraphics[width=2.9in,height=1.65in]{./figures/asym_compression.eps}
%     \caption{Asymptotic regime $N_s \rightarrow \infty$ with $\widehat{\R}_y = \R_\y$}
%     \label{fig:error_asymp}
%               \end{subfigure}               
%\\[1em]
~
	\begin{subfigure}[b]{0.4\textwidth}
		     \centering
                    \includegraphics[width=3in,height=1.65in]{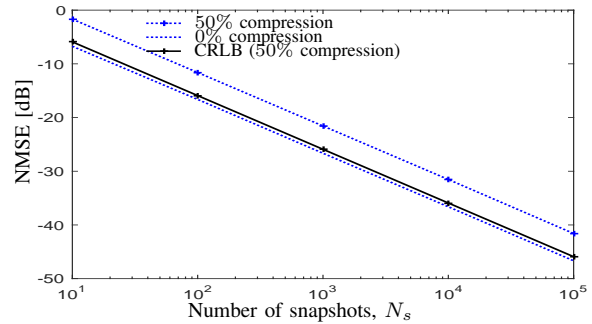}
                              \caption{Non-parametric method}
               \label{fig:mse_nonpar}
 \end{subfigure} 
 \\[1em]

	\begin{subfigure}[b]{0.4\textwidth}
	     \centering	
               \includegraphics[width=3in, height=1.65in]{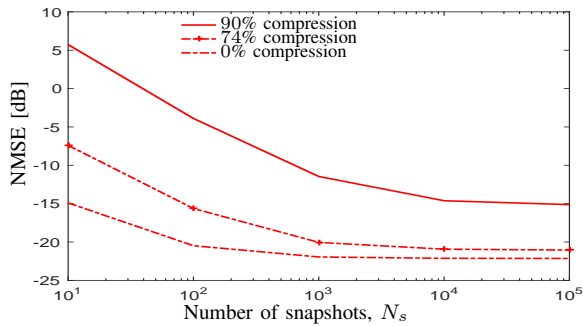}
               \caption{Moving average model ($Q=13$)}
               \label{fig:mse_ma}
                              \end{subfigure} 
 ~\hskip1.75cm
     \begin{subfigure}[b]{0.4\textwidth}
          \centering
     \includegraphics[width=3in,height=1.65in]{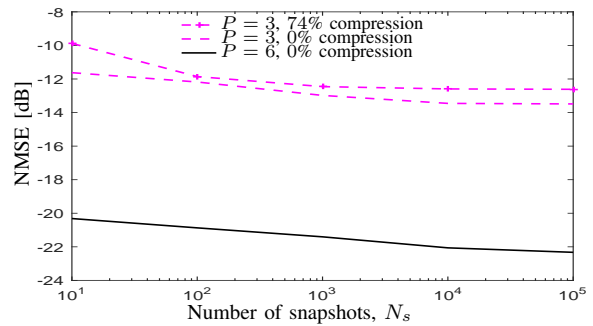}
     \caption{Autoregressive model}
     \label{fig:mse_ar}
	\end{subfigure} 
     \caption{Performance analysis on the synthetic dataset. In (a), markers indicate the non-uniformly distributed eigenvalues of the graph Laplacian matrix along the x-axis. }
     \label{fig:mse}
%Graph power spectrum based on $N_s=1000$ snapshots. Markers indicate the non-uniformly distributed eigenvalues of the graph Laplacian matrix along the x-axis. (a) NMSE for different snapshots for the non-parametric approach. (b) NMSE for different snapshots for the moving average model with $Q=13$. (c) NMSE for different snapshots for the autoregressive model.     
\end{figure*}
\section{Numerical Experiments} \label{sec:numericalresults}

The developed framework of sampling on graphs for power spectrum estimation is illustrated with numerical experiments\footnote{\noindent Software and datasets to reproduce results of this paper can be downloaded from  {\url{http://cas.et.tudelft.nl/~sundeep/sw/jstsp16gpsd.zip}}} on synthetic as well as real datasets as discussed next. 

\subsubsection*{Synthetic data (random graph)} For experiments using synthetic data, a random sensor graph with $N=100$ nodes is generated using the {\tt GSPBOX}~\cite{perraudin2014gspbox}. The generated graph topology can be seen in Figure~\ref{fig:nodesamp_synthetic}, where the colored nodes represent the value of the graph signal for one realization. Graph stationary signals are generated by graph filtering zero-mean unit-variance white noise with a filter, which has a squared magnitude frequency response as shown in Figure~\ref{fig:psd_syn} (labeled as ``True graph power spectrum");  such a frequency response can be, for instance, { approximated} using a filter with $L=7$ coefficients. {For the shift operator, we use the graph Laplacian matrix.} We use $N_s = 1000$ snapshots to form a sample covariance matrix, which we use in the experiments.  

%Figure~\ref{fig:psd_syn} also shows the reconstructed graph power spectrum for the different approaches considered in this paper. For the non-parametric approach, we observe a subset of $K=50$ out of $N=100$ (i.e., with $50\%$ compression). The observed nodes are circled as can be seen in Figure~\ref{fig:nodesample1_nonpar}, however, no particular sampling pattern can be seen. 

\begin{figure*}[!t]
\psfrag{index}{\hskip-1mm\scriptsize index}
\psfrag{singular value}{\hskip-2mm\scriptsize singular value}
\psfrag{minimum sparse ruler}{\scriptsize min. sparse ruler}
\psfrag{submodular design}{\scriptsize submodular design}
     \centering
   \begin{subfigure}[b]{0.4\columnwidth}  
   \includegraphics[width=\columnwidth,height=1.25in]{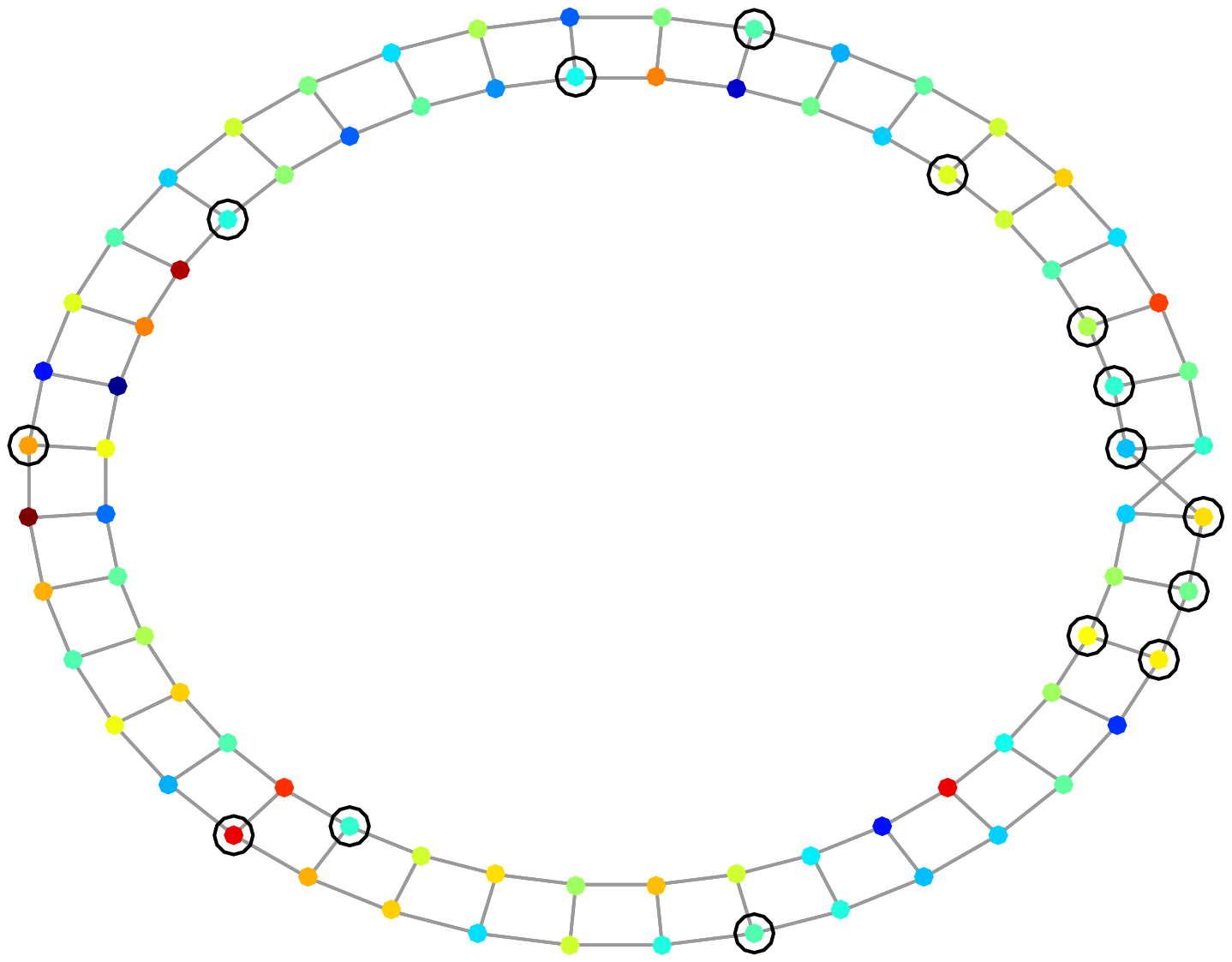}
   \caption{}
   \label{fig:minspars_mobius}
   \end{subfigure}
~\hskip1.75cm
   \begin{subfigure}[b]{0.4\columnwidth}
   \includegraphics[width=\columnwidth,height=1.25in]{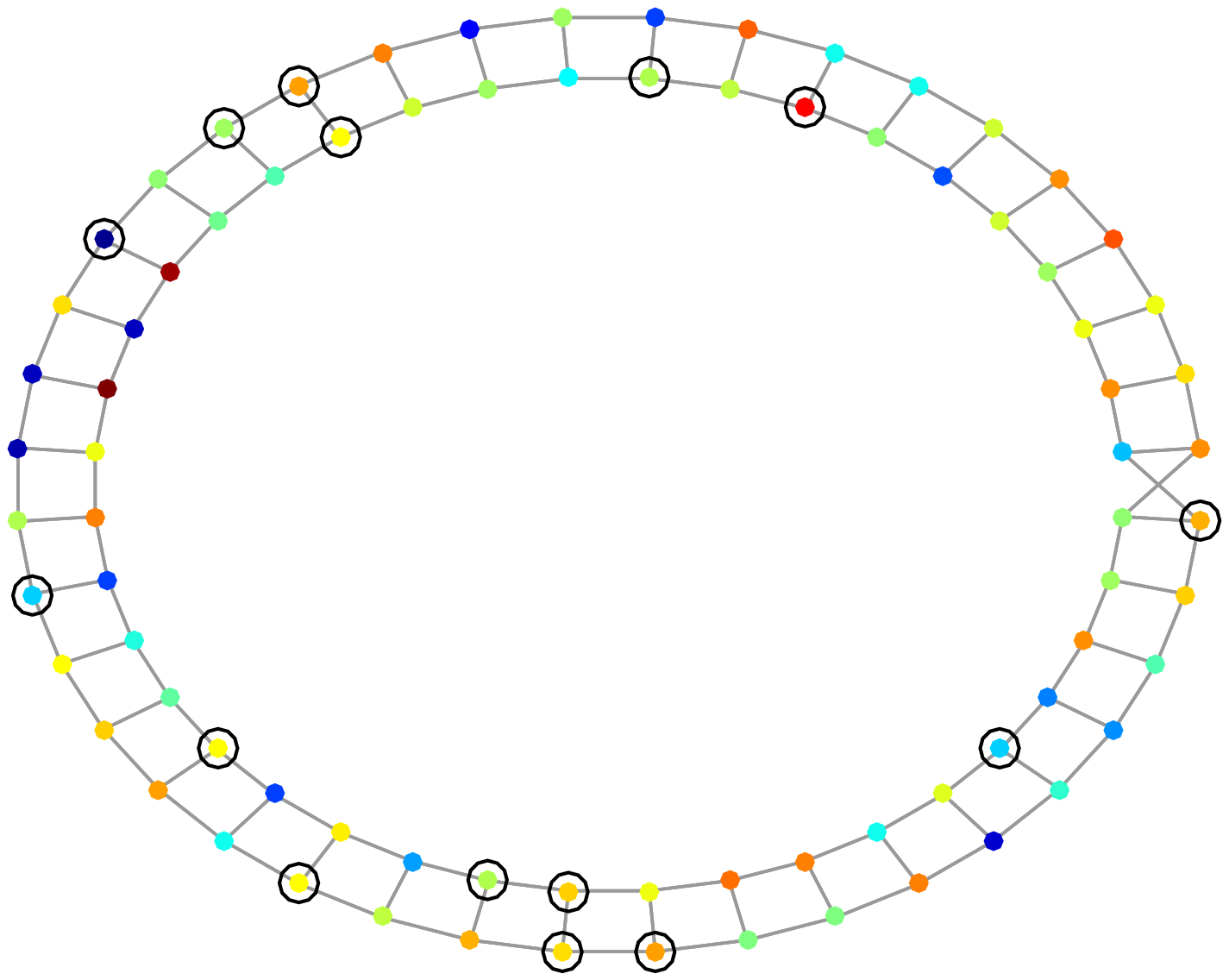}
      \caption{}
     \label{fig:submod_mobius}
\end{subfigure}
~\hskip1cm
        \begin{subfigure}[b]{0.65\columnwidth}
        \includegraphics[width=\columnwidth,height=1.3in]{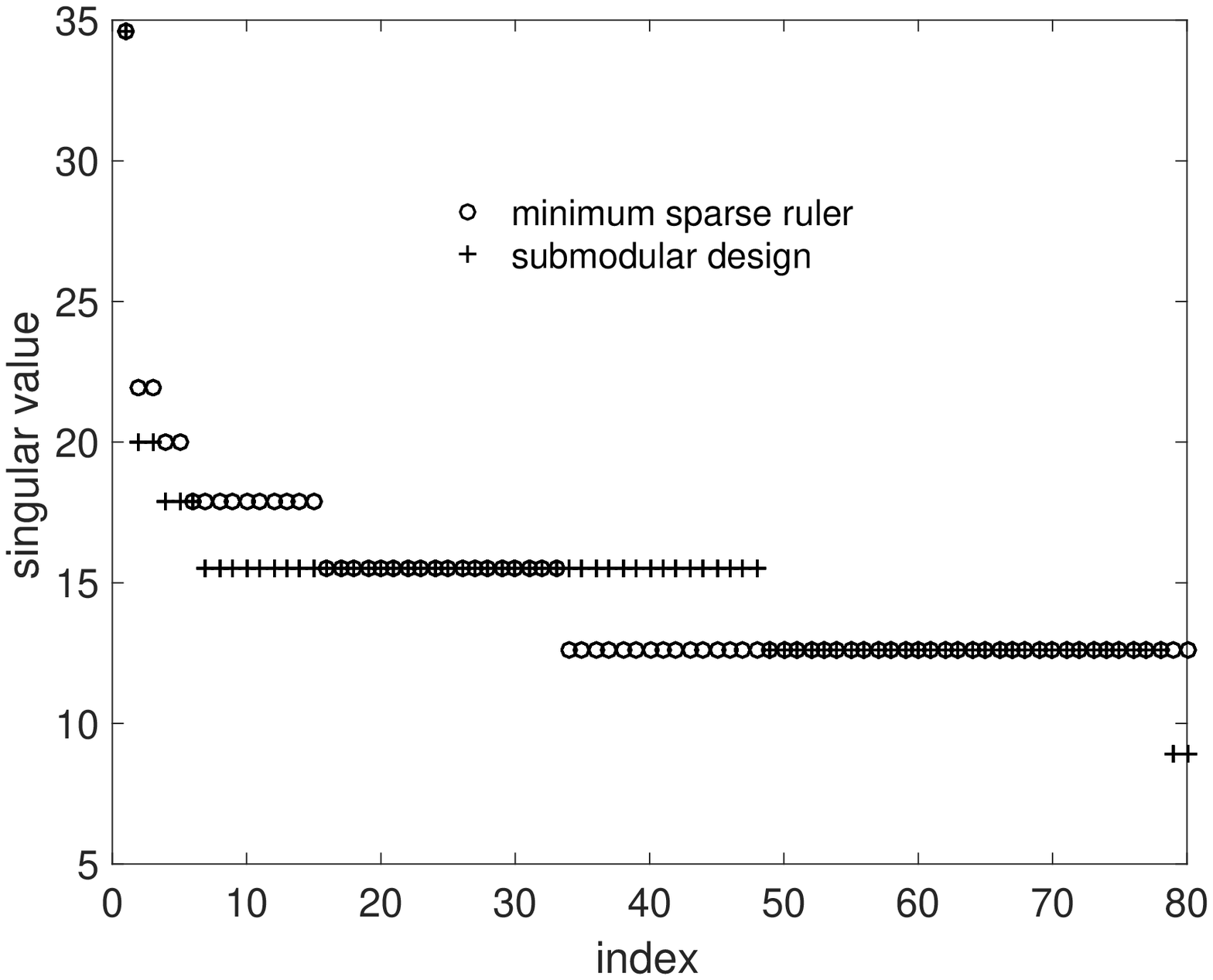}
           \caption{}
        \label{fig:svd_submod}
        \end{subfigure}
     \caption{Sampling M\"obius ladder \textemdash a circulant graph with $N=80$ nodes. The sampled graph nodes are highlighted by the circles around the nodes and the node coloring simply denotes a realization of the graph signal. (a) Minimal sparse ruler based sampling with $K = 15$. (b) Sampling based on submodular design with $K = 15$. (c) Spectrum of ${\boldsymbol T}({\boldsymbol w}) = {\boldsymbol \Psi}_{\rm s}^T ({\rm diag}[{\boldsymbol w}] \otimes {\rm diag}[{\boldsymbol w}]) {\boldsymbol \Psi}_{\rm s}$ with ${\boldsymbol w}$ being the minimal sparse ruler and for ${\boldsymbol w}$ computed using the greedy submodular design.}
     \label{fig:mobius_sparseruler}
     \vspace*{-2mm}
\end{figure*}

For the non-parametric model, using Algorithm~\ref{alg:greedy}, we first design the subsampler by selecting rows of the matrix ${\boldsymbol \Psi}_{\rm s}$ in a structured manner determined by $\w$. We show in Figure~\ref{fig:psd_syn}, that the least squares estimate of the graph power spectrum obtained by observing $K=50$ out of $N=100$ nodes ($50\%$ compression) fits reasonably well to the true power spectrum. In Figure~\ref{fig:nodesample1_nonpar}, the selected graph nodes are indicated with a black circle. However, no particular sampling pattern can be seen here.  
%That is, we perform a row subset selection of the (modified) graph Fourier matrix ${\bar{\boldsymbol U}} \circ {\boldsymbol U}$. 

For the parametric moving average model, recall that the graph power spectrum is parameterized with $Q$ parameters; we use $Q=13$ in this example. As before, we perform a row subset selection of the matrix ${\boldsymbol \Psi}_{\rm MA}$ in a structured manner using Algorithm~\ref{alg:greedy}. We show in Figure~\ref{fig:psd_syn}, the (unconstrained) least squares estimate of the graph power spectrum computed using observations from $K=26$ nodes out of $N=100$ nodes ($74\%$ compression). The sampling pattern in this case is shown in Figure~\ref{fig:nodesample1_ma}. It can be seen that the greedy algorithm selects graph nodes in a clustered manner  { as the moving average model assumes that the power spectrum is smooth. }

For the parametric autoregressive approach, the graph power spectrum is parameterized with $P =3$ parameters. In this case, we choose $K_0 = 1$ graph node (indicated with a red circle) having the largest degree and we also observe nodes in the $3$-hop neighborhood of the selected node; the observed nodes (indicated with black circles) are shown in Figure~\ref{fig:nodesample1_ar}. In this example, we observe $K=26$ nodes out $N=100$ nodes to reconstruct the graph power spectrum. The least squares estimate of the G-AR power spectrum can be seen in Figure~\ref{fig:psd_syn}. {Although we had to recover only $P=3$ parameters, we observe all the nodes in the $P$-hop neighborhood of every selected node (i.e., we observe much more than $K_0P$ nodes).} 
%This means that some of the nodes are redundant, and thus may be dropped to reduce the processing overhead.}
%
%%

\begin{figure*}
     \psfrag{celcius}{\tiny \hskip-2mm $^\circ$ Celcius}
\psfrag{Laplacian eigenvalues}{\hskip-10mm\footnotesize Shift-operator (adjacency) eigenvalues}
\psfrag{non-parametric approach 20 nodes}{ \scriptsize Non-parametric ($K=20$ nodes)}
\psfrag{True PSD}{ \scriptsize Empirical (uncompressed)}
\psfrag{Estimated PSD autoreggresive model 9 nodes}{\scriptsize Autoregressive ($K=9$ nodes)}
\psfrag{Estimated PSD moving average model 20 nodes}{\scriptsize Moving average ($K=20$ nodes)}
     \psfrag{dB}{\tiny dB}
     \centering
\begin{subfigure}[b]{0.52\columnwidth}
	     \centering	
	     \includegraphics[width=\columnwidth,height=1.5in]{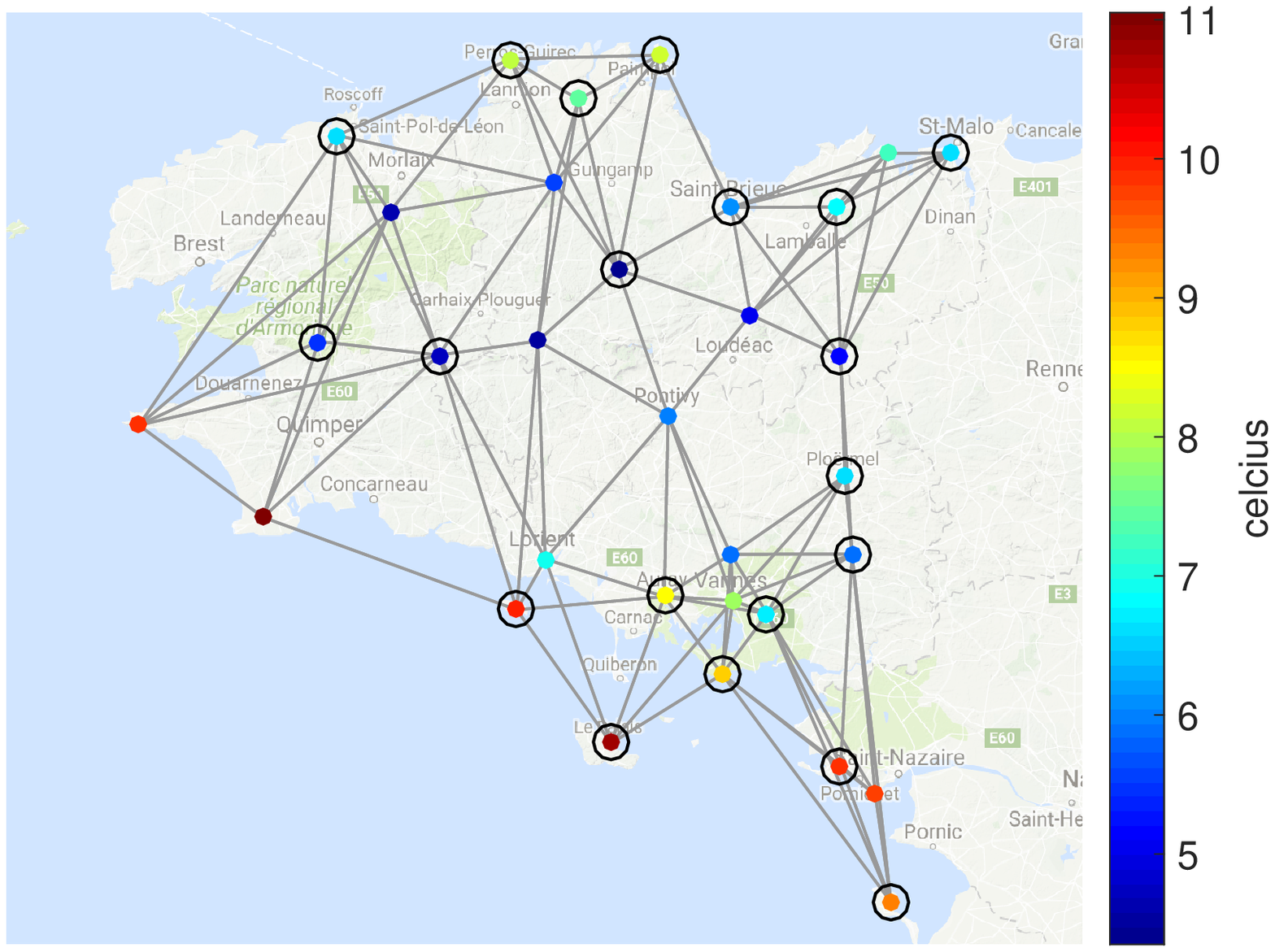}
	     \caption{}
	     \label{fig:nodesample2_nonpar}
	     \end{subfigure}
~\hskip1cm
\begin{subfigure}[b]{0.52\columnwidth}
	     \centering	
\includegraphics[width=\columnwidth,height=1.5in]{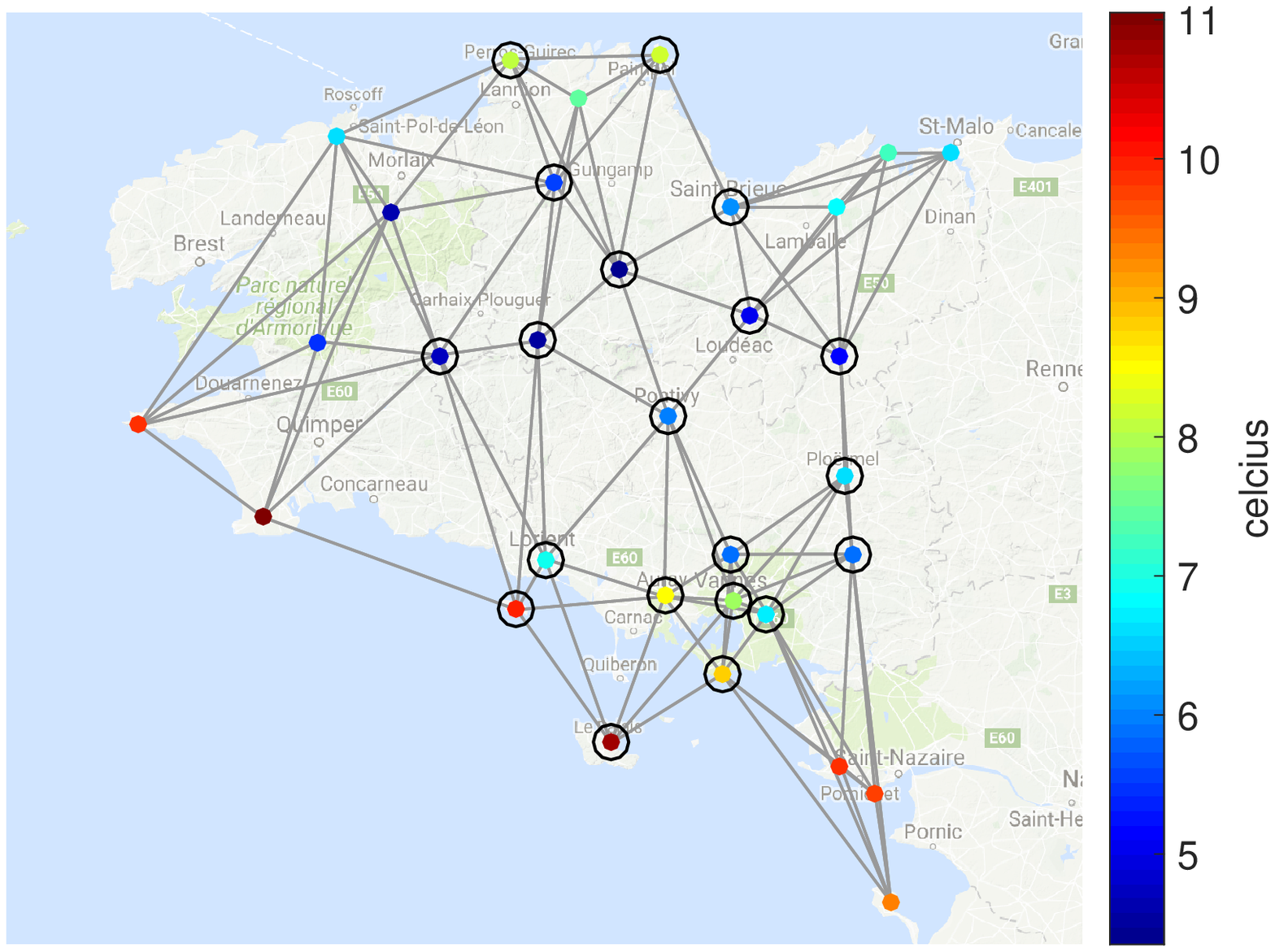}
	     \caption{}
\label{fig:nodesample2_ma}
\end{subfigure}
~\hskip1cm
\begin{subfigure}[b]{0.52\columnwidth}
	     \centering	
\includegraphics[width=\columnwidth,height=1.5in]{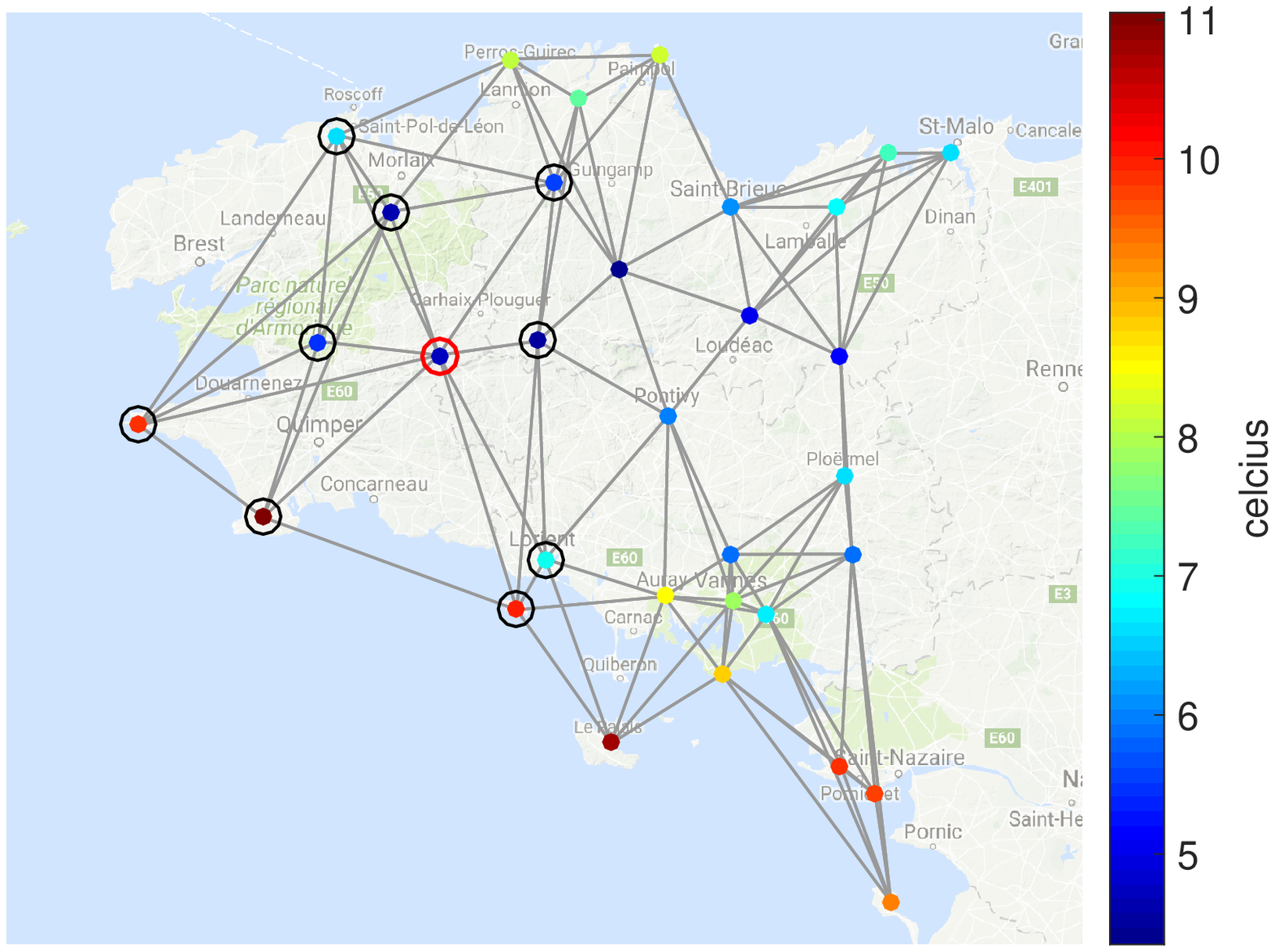}
	     \caption{}\label{fig:nodesample2_ar}
	     \end{subfigure}
\\[1.5em]
\begin{subfigure}[b]{0.45\textwidth}
	     \centering	
\includegraphics[width=0.7\columnwidth, height=1.75in]{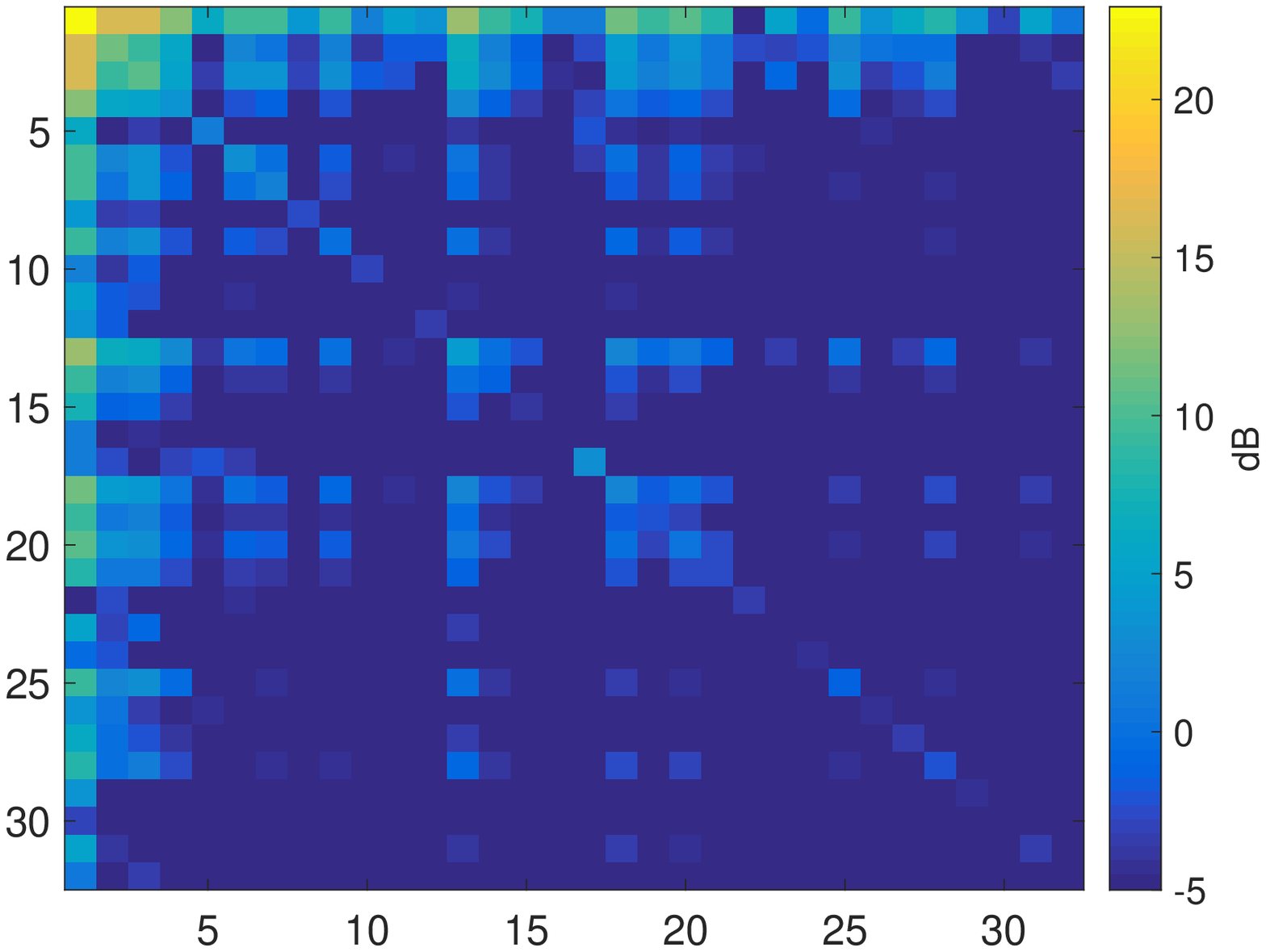}
	     \caption{}\label{fig:spectral_covariance}
	     \end{subfigure}	     
~%\hskip1cm
\begin{subfigure}[b]{0.45\textwidth}
\centering
                \includegraphics[width=0.7\columnwidth, height=1.75in] {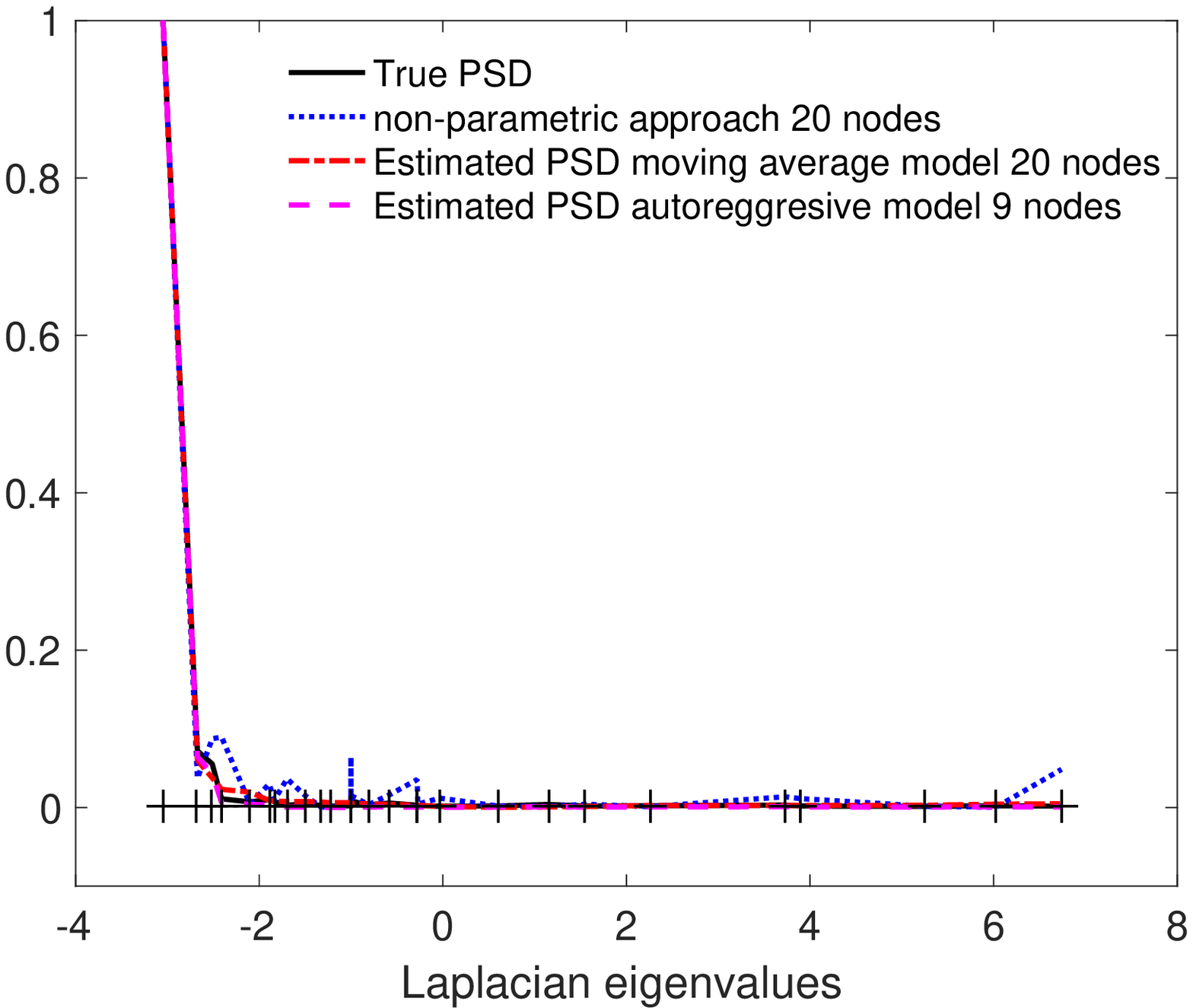}
                \caption{}
\label{fig:psd_real}
\end{subfigure}	     
     \caption{Sampling on graphs with $N=32$ weather stations. The sampled graph nodes are highlighted by the circles around the nodes and the node coloring simply denotes a realization of the graph signal. (a) Non-parametric model with $K = 20$. (b) Moving average model with $Q=11$ and $K = 20$. (c) Autoregressive model with $K_0=1$, where the $P$-hop neighborhood around the node indicated with the red circle is observed. (d) Spectral covariance matrix. (e) Graph power spectrum based on $N_s = 744$ snapshots. Markers along the x-axis indicate the eigenvalues of the adjacency matrix}
     \label{fig:nodesamp_real}
\end{figure*}
{In Figure~\ref{fig:mse}, we also provide some performance results based on the synthetic dataset. In particular, we show for different number of snapshots the performance of the estimators in terms of the normalized mean squared error (NMSE) defined in dB as 
$
{\rm NMSE} = 10 \log_{10} \, \sum_{m=1}^{\rm N_{\rm exp}} \| \p - \widehat{\p}_m\|_2^2/ (N_{\rm exp}\|\p\|_2),
$
where $\widehat{\p}_m$ denotes the graph power spectrum estimate during the $m$th Monte-Carlo experiment and $N_{\rm exp}$ is the number of Monte-Carlo experiments. Here, we use $N_{\rm exp} = 1000$.  

{
To begin with, Figure~\ref{fig:mse_nonpar} shows the performance of the developed least squares estimator for the nonparametric approach with $K=50$ (50$\%$ compression), and with $K=100$, i.e., no compression. For this example, we can see about a $4$~dB performance loss due to compression, and this gap reduces as $K$ increases. Furthermore, we can also see that, although the least squares estimator has the same slope as that of the Cram\'er-Rao lower bound (labeled as ``CRLB (50$\%$ compression)"), it does not achieve the Cram\'er-Rao lower bound. This gap can be reduced by solving a weighted least squares estimator, but incurs an additional computational cost due to inverting and updating the weighting matrix. 
For this particular scenario, although a full-column rank matrix $({\boldsymbol \Phi} \otimes {\boldsymbol \Phi}){\boldsymbol \Psi}_{\rm s}$ can be obtained for $K \geq 20$, but $K=20$ results in a very poor performance as ${\boldsymbol \Psi}_{\rm s}$ is highly sensitive to perturbations due to the finite sample effects. Nevertheless, the performance improves with the number of snapshots.}
%We can see this in Figure~\ref{fig:error_asymp}, where we show the performance, for different compression rates, in the asymptotic regime, i.e., as $N_s \rightarrow \infty$, for which we simply use $\widehat{\R}_\y = \R_\y$.

In Figure~\ref{fig:mse_ma}, we can see the performance of the moving average approach for $Q=13$, for $K=10$ (90$\%$ compression, which is also the maximum possible compression for this example),  $K = 26$ (74$\%$ compression) and $K = 100$ (i.e., no compression). As before, we see a performance loss due to compression, but also, as the number of snapshots increases, the performance saturates. This is due to the limited filter order, and the performance gets better with increasing filter order. However, increasing the filter order worsens the condition number of $\Psib_{\rm MA}$, and we might have to resort to singular value decomposition based techniques to solve the least squares problem (now we simply solve \eqref{eq:ls_finite} using QR factorization technique through MATLAB's {\it backslash} ``\textbackslash"~operator). For this example, a full-column rank matrix $({\boldsymbol \Phi} \otimes {\boldsymbol \Phi}){\boldsymbol \Psi}_{\rm MA}$ is obtained for $K \geq 10$. Such a high compression is possible because of the low value of $Q$ that is assumed to be known. Also, as compared to the non-parametric model, due to a smaller filter order, ${\boldsymbol \Psi}_{\rm MA}$ is less sensitive to perturbations. This can be see in Figure~\ref{fig:mse_ma}, where we get a reasonable performance for the maximum possible compression with $K=10$.
%, and the performance converges to that of the asymptotic regime shown in Figure~\ref{fig:error_asymp}, where we show the %asymptotic performance of the moving average model for different compression rates.

Finally, in Figure~\ref{fig:mse_ar}, we show the performance of the autoregressive model for $P=3$ with $K =1$ and $K=100$, and for $P=6$ with $K=100$ we solve \eqref{eq:AR_uncomp} using least squares. Although we can see a similar behavior with respect to the performance loss due to compression and with respect to the error saturation due to a limited filter order, a more important thing to notice is that the autoregressive model has a similar performance as that of the moving average model, but with about $50\%$ fewer parameters.
}

{\subsubsection*{Synthetic dataset (circulant graph)} We illustrate the graph sampling theory developed for circulant graphs using a M\"obius ladder, which due to its structure finds applications within molecular chemistry (e.g., see~\cite{flapan1989symmetries}). A M\"obius ladder with $N=80$ nodes is shown in Figure~\ref{fig:minspars_mobius}. This graph has a circulant adjacency matrix, which we use as the shift operator. 

We have seen in Section~\ref{sec:cyclicgraphs} that for such circulant graphs it is possible to elegantly compute the optimal sparse samplers. For $N=80$, the minimal sparse rulers are length $K=15$ and one such (non-unique) sampling set is given by $\mathcal{K}= \{1, 2,  3,  6,  11,  16,    27, 38, 49, 60, 66, 72, 78,79,80\}$; see the corresponding selected nodes in Figure~\ref{fig:minspars_mobius}. Alternatively, we can also determine the sampling set using Algorithm~\ref{alg:greedy}; we show the selected nodes in Figure~\ref{fig:submod_mobius}. Now, the question is, how does this greedily designed sparse sampler compare with the minimal sparse ruler. To answer this, we plot, in Figure~\ref{fig:svd_submod} the singular values (i.e., the spectrum) of ${\boldsymbol T}({\boldsymbol w}) = {\boldsymbol \Psi}_{\rm s}^T ({\rm diag}[{\boldsymbol w}] \otimes {\rm diag}[{\boldsymbol w}]) {\boldsymbol \Psi}_{\rm s}$ with ${\boldsymbol w}$ being the minimal sparse ruler and for ${\boldsymbol w}$ computed using the greedy submodular design. For this example, we can see the resulting spectrum from both the sparse samplers are very similar, and that the greedy submodular design has a slightly worse condition number (i.e., the ratio of maximal singular value to minimal singular value).
}
\begin{figure*}
\psfrag{Laplacian eigenvalues}{\hskip-12mm\footnotesize Shift-operator (Laplacian) eigenvalues}
\psfrag{Empirical PSD}{ \scriptsize Empirical (uncompressed)}
\psfrag{dB}{ \tiny dB}
\psfrag{Estimated PSD moving average model 15 nodes}{\scriptsize Moving average ($K=15$)}
\psfrag{Estimated PSD moving average model 256 nodes}{\scriptsize Moving average ($K=256$)}
     \centering
\begin{subfigure}[b]{0.52\columnwidth}
     \centering
\includegraphics[width=\columnwidth,height=1.5in]{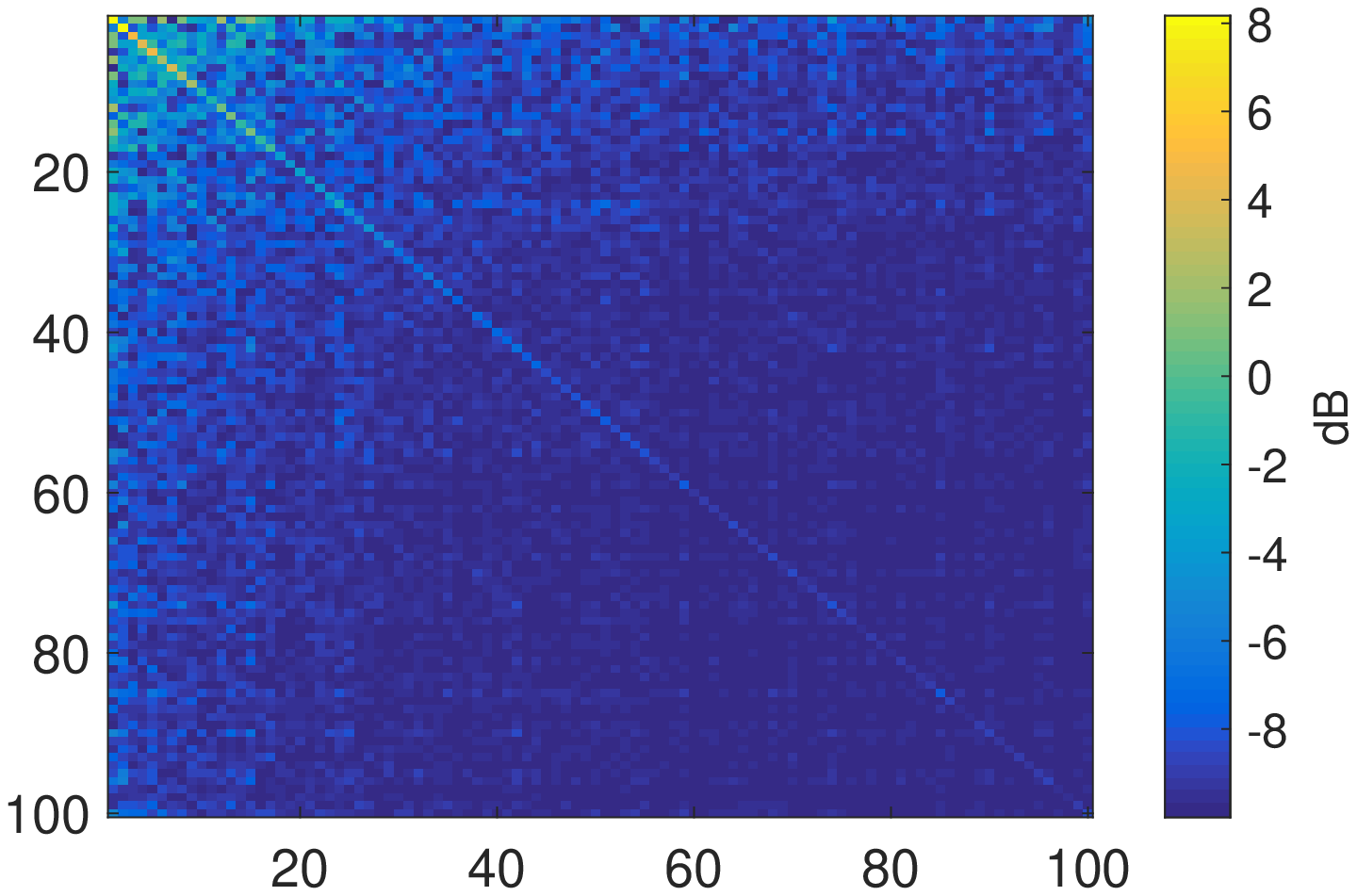}
\caption{}
\label{fig:usps_speccov}
        \end{subfigure}
~\hskip1cm
\begin{subfigure}[b]{0.52\columnwidth}
     \centering
     \includegraphics[width=\columnwidth,height=1.5in]{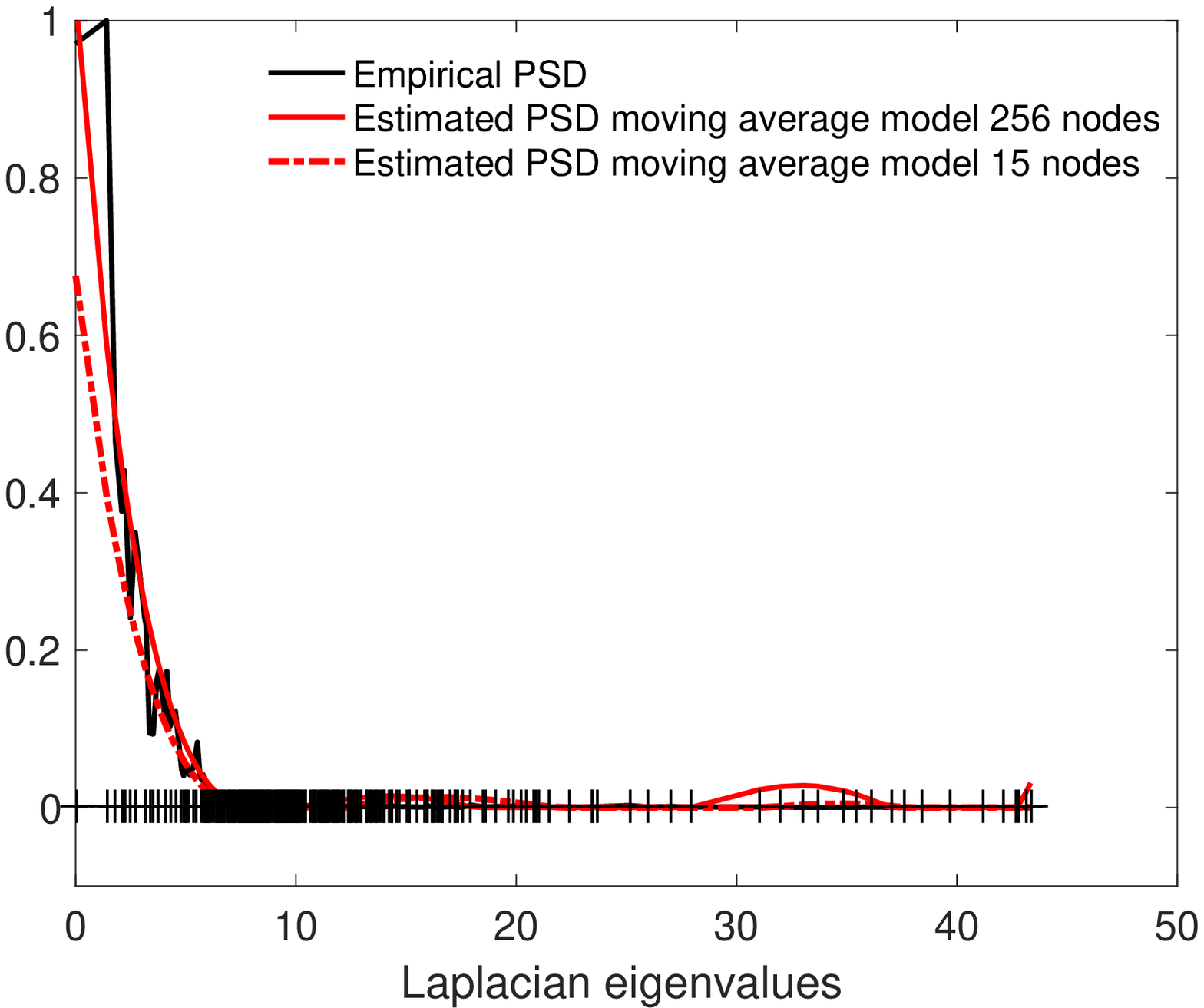}
     \caption{}
     \label{fig:usps_psd}
             \end{subfigure}
~\hskip1.25cm
\begin{subfigure}[b]{0.45\columnwidth}
     \centering
\includegraphics[width=\columnwidth,height=1.5in]{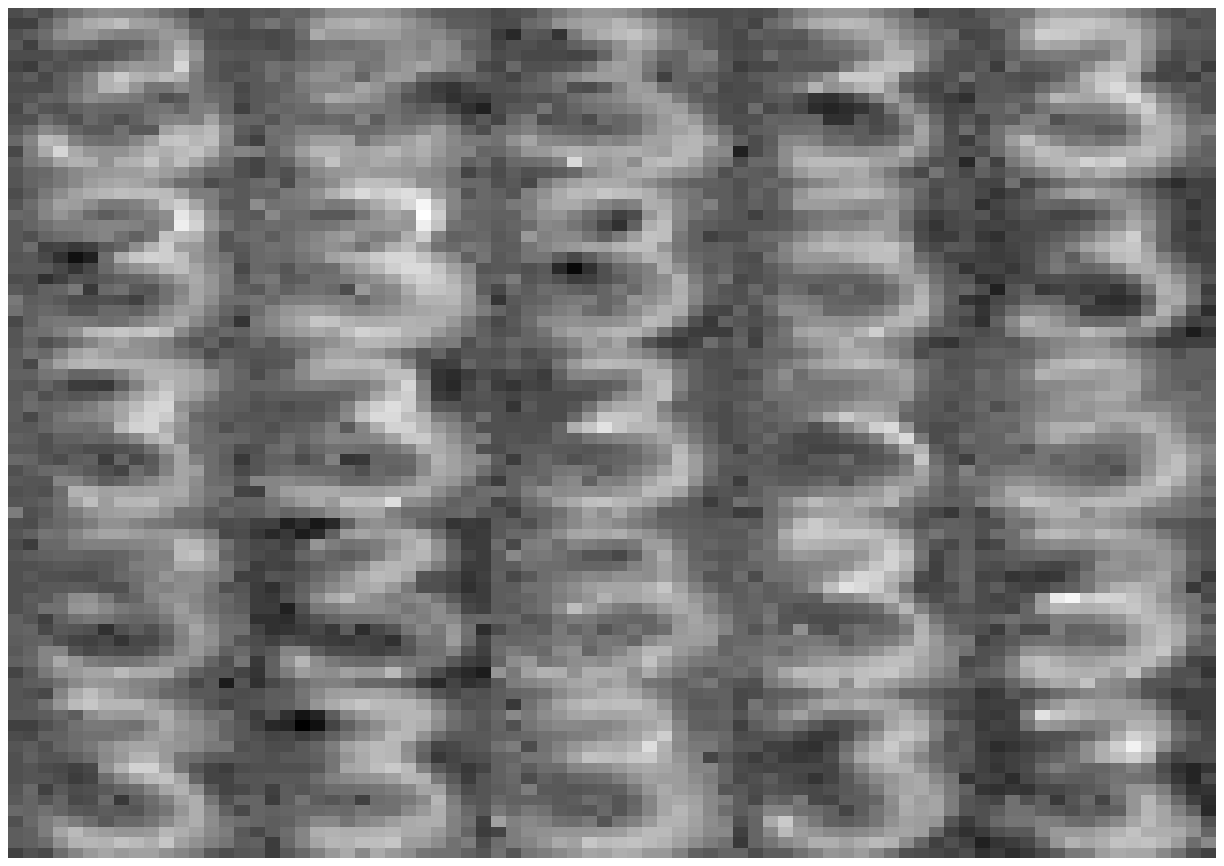}
\caption{}
\label{fig:usps_gen}
             \end{subfigure}
     \caption{Sampling nearest neighbor graph built using digit 3 ($16 \times$ 16 pixels) from the USPS dataset. (a) Spectral covariance matrix (only the upper left part is shown for better visibility, rest of the entries are zeroes). (b) Graph power spectrum based on $N_s = 50$ image snapshots. Markers along the x-axis indicate the eigenvalues of the Laplacian matrix. (c) 25 realizations of the generated images, which are obtained by graph filtering white noise. Here, the $Q=7$ G-MA filter coefficients are obtained by observing $K=15$ pixels.}
     \label{fig:USPS}
\end{figure*}

\subsubsection*{Real dataset (climatology)} For the real dataset, we use temperature measurements collected across $32$ different weather stations in the French region of Brittany\footnote{{This dataset was used in the context of stationary graph signal processing in~\cite{benjamin15eusipco,perraudin2016stationary}. Also, we would like thank the authors of~\cite{perraudin2016stationary} for making this as well as the USPS (preprocessed) datasets public.}}. A nearest neighbor graph is constructed as in~\cite{perraudin2016stationary} using the available coordinates of the weather station { such that each node has at least five neighbours.}  The reconstructed graph can be seen in Figure~\ref{fig:nodesamp_real}. Alternatively, the method suggested in~\cite{chepuri2016learning} can be used to construct a sparse graph based on training data. There are $N_s = 744$ observations (for 31 days and 24 observations per day) per weather station available. {We use the adjacency matrix as the shift operator in this example.}

{ We have removed the (sample) mean from each station independently, thus forcing the first moment to zero~\cite{perraudin2016stationary}. This way we artificially obtain ${\boldsymbol m}_\x = m_\x \u_1$ with $m_\x = 0$. After removing the mean, the temperature data records are nearly stationary on this graph, i.e., the sample covariance matrix (denoted by $\widehat{\R}_\x$) in the graph spectral domain (i.e., the spectral covariance matrix $\U \widehat{\R}_\x\U$) has most of its energy, i.e., about $89\%$ of the energy of $\U \widehat{\R}_\x\U$, along the main diagonal; see the spectral covariance in Figure~\ref{fig:spectral_covariance}. The stationarity of this dataset on the shift operator increases when processing the so-called intrinsic mode functions of the temperature recordings instead of the raw data as detailed in~\cite{girault2015signal}, but we will simply use the mean-removed raw dataset here.}

We carry out the same experiments as for the synthetic data. For the non-parametric and moving average approaches, the samplers are designed using a greedy algorithm as discussed in Section~\ref{sec:grdy_manonpar}. In particular, for the non-parametric approach, we observe $K=20$ nodes out of $N=32$ nodes as shown with black circles in Figure~\ref{fig:nodesample2_nonpar}. For the moving average approach, we use $Q=11$, and observe $K=20$ out of $N=32$ nodes to recover the G-MA parameters.  Finally, for the autoregressive approach, we model the graph power spectrum with $P=1$ scalar parameter. We select one node (i.e., $K_0=1$) that has the largest degree as indicated with a red circle in Figure~\ref{fig:nodesample2_ar}, and we also observe nodes in the one-hop neighborhood of the selected node. So, we observe 9 nodes in total in this case. The uncompressed graph power spectrum computed from all the available temperature measurements as well as the least squares estimate of the graph power spectrum computed from the subsampled observations using the non-parametric and parametric approaches can be seen in Figure~\ref{fig:psd_real}, where we can see that the shape of estimated power spectrum from different approaches is similar to that of the empirical graph power spectrum.

{
\subsubsection*{Real dataset (USPS handwritten digits)}
Before concluding, we will demonstrate the potential of parametric modeling as well as sampling in the graph setting with an example using the USPS dataset, where we will focus only on {\it digit 3} for the sake of illustration. We construct a 20 nearest neighbor graph with 50 images each containing $16 \times 16$ pixels as in~\cite{perraudin2016stationary}. This means that the graph signal ${\x}$ is of length $256$, where each pixel corresponds to a graph node, and the covariance matrix $\R_\x$ is of size $256 \times 256$. The stationarity of this dataset on such a graph has been demonstrated in~\cite{perraudin2016stationary}; see the diagonal dominance (with about $82\%$ of the energy in the diagonal entries) of the spectral covariance matrix in Figure~\ref{fig:usps_speccov}. 

We have seen in Section~\ref{sec:parametric} that it is possible to model the graph power spectrum with fewer parameters, which means that (a) we need to store or transmit only a few parameters, and (b) we can achieve stronger compression rates. To illustrate this, we perform an experiment, where we view {\it digit 3} of the USPS dataset as a realization of a graph second-order stationary signal obtained by graph filtering white noise using a graph moving average filter with $Q=7$.  In Figure~\ref{fig:usps_psd}, we show the empirical graph power spectrum computed from $50$ images and the graph power spectrum computed using the moving average method by sampling only $K=15$ pixels (96$\%$ compression) as well as $K=256$ (i.e., no compression). That is to say, we can quickly learn the parameters of interest without visiting the entire training set. Next, based on the reconstructed graph power spectrum obtained by sampling $K=15$ pixels, we generate $25$ realizations of graph signals by graph filtering white noise, where the frequency response of the graph filter is simply computed as $h_{f,n} = |p_n|^{1/2}$ for $n=1,\ldots,N$ (here, we use the absolute value because we do not solve \eqref{eq:ls_finite} with a nonnegativity constraint). These 25 realizations are shown in Figure~\ref{fig:usps_gen}, where we can see that the resulting signals have the shape of {\it digit 3} corroborating that the signal is stationary on the nearest neighbor graph, and more importantly these signals can be generated from fewer parameters, which are estimated by observing only a small subset of pixels.
%In sum, we have demonstrated that by observing only a few nodes and using simple least squares we can recover the shape of the graph power spectrum reasonably well. 

}

%\subsection{Numerical example} \label{sec:numerical_example}

\section{Concluding Remarks} \label{sec:conc}
In this paper we have focused on sampling and reconstructing the second-order statistics of stationary graph signals. The main contribution of the paper is that by observing a significantly smaller subset of vertices and using simple least squares estimators, we can reconstruct the second-order statistics of the graph signal from the subsampled observations, and more importantly, without any spectral priors. The results provided here generalize the compressive covariance sensing framework to the graph setting. Both a nonparametric approach as well as parametric approaches including moving average and autoregressive models for the graph power spectrum are discussed. A near-optimal low-complexity greedy algorithm is developed to design a sparse sampling matrix that selects the subset of graph nodes. 
\vspace*{-5mm}
\appendices
\section{Lemma~\ref{lem:Psifullrank}: Rank of self Khatri-Rao products} \label{app:proofselfkr}

By the definition in \eqref{eq:graphShift}, $\U$ forms an orthogonal basis and hence full rank. As a result, the sum 
$
a_1\u_1 +a_2\u_2+ \cdots +a_N\u_N 
$
equals zero only when $a_1=a_2=\cdots=a_N=0$.

The remainder of the proof is based on contradiction. Assume that the matrix $\bar{\U} \circ \U = [\bar{\u}_1 \otimes \u_1, \cdots, \bar{\u}_N \otimes \u_N]$ does not have full column rank. This means that the sum 
\begin{equation}
\begin{aligned}
&b_1(\bar{\u}_1 \otimes \u_1) +  \cdots + b_N(\bar{\u}_N \otimes \u_N)\\
&= b_1 \left[\begin{array}{c}\bar{u}_{1,1}\u_1 \\ \vdots \\\bar{u}_{1,N}\u_1  \end{array}\right] + \cdots  + b_N\left[\begin{array}{c}  \bar{u}_{N,1}\u_N \\ \vdots \\  \bar{u}_{N,N}\u_N\end{array}\right] = {{\bf 0}}
\end{aligned}
\end{equation}
when one or more $b_i \bar{u}_{i,j}$ are nonzero. This is possible only if $\U$ is singular. Hence a contradiction, implying that ${\rm rank}(\bar{\U} \circ \U) = N$. 

%\section{Lemma~\ref{lem:rankKron}: Rank of Self Kronecker products} \label{app:proofselfkron}
%
%The proof follows from the singular value decomposition of $\Phib \otimes \Phib$ and the property $({\boldsymbol A} \otimes {\boldsymbol B}) ({\boldsymbol C} \otimes  {\boldsymbol D})  = ({\boldsymbol A}{\boldsymbol C} \otimes {\boldsymbol B}{\boldsymbol D})$.
%
%The proof follows from the singular value decomposition of $\Phib \otimes \Phib$. Let $\Phib$  have a singular value decomposition $\Phib = \U_\Phib {\boldsymbol \Sigma}_\Phib {\boldsymbol V}_\Phib^T$. Then using the property $({\boldsymbol A} \otimes {\boldsymbol B}) ({\boldsymbol C} \otimes  {\boldsymbol D})  = ({\boldsymbol A}{\boldsymbol C} \otimes {\boldsymbol B}{\boldsymbol D})$, we have 
%$
%\Phib \otimes \Phib = [\U_\Phib \otimes \U_\Phib] [{\boldsymbol \Sigma}_\Phib \otimes {\boldsymbol \Sigma}_\Phib] 
%[{\boldsymbol V}_\Phib \otimes {\boldsymbol V}_\Phib]^H.
%$ 
%This implies, ${\rm rank}({\boldsymbol \Phi} \otimes {\boldsymbol \Phi}) = [{\rm rank}({\boldsymbol \Phi})]^2 = K^2$.
%
%
\vspace*{-5mm}
\section{Theorem~\ref{theo:fullrank_spectral}: Conditions for a Valid Sampler} \label{app:theo1}

The rank of the product of two matrices $\A$ and $\B$ is given by~\cite{meyer2000matrix} 
$
{\rm rank}(\A\B) \leq  \min\{{\rm rank}(\A),{\rm rank}(\B)\}, 
$
and equality holds if and only if ${\rm null}(\A) \cap {\rm ran}(\B) = \emptyset$.
%- {\rm dim}({\rm null}(\A) \cap {\rm ran}(\B) ),
%\[
%{\rm rank}(\A\B) =  {\rm rank}(\B) - {\rm dim}({\rm null}(\A) \cap {\rm ran}(\B) ),
%\]
%which means, ${\rm rank}(\A\B) =  {\rm rank}(\B)$ if and only if ${\rm null}(\A) \cap {\rm ran}(\B) = \emptyset$. 

%Therefore, rank of $({\boldsymbol \Phi} \otimes {\boldsymbol \Phi}) {\boldsymbol \Psi}_{\rm s}$ is $\min\{K^2, N\}$, where 

We know from Lemma~\ref{lem:rankKron} that ${\rm rank}({\boldsymbol \Phi} \otimes {\boldsymbol \Phi})$ is $K^2$ if ${\rm rank}({\boldsymbol \Phi}) = K$ and from Lemma~\ref{lem:Psifullrank} that $\Psib_{\rm s}$ has full column rank. This implies that if $K^2 \geq N$, then $({\boldsymbol \Phi} \otimes {\boldsymbol \Phi}) {\boldsymbol \Psi}_{\rm s}$ has full column rank provided that the null space of $\Phib \otimes \Phib$ (which is generated by the basis vectors in the null space of $\Phib$) does not intersect with the space spanned by the columns of $\Psib_{\rm s}$.

% -------------------------------------------------------------------------
%\vfill
%\newpage
%\bibliographystyle{IEEEbib}
%
\bibliographystyle{IEEEtran}
\bibliography{IEEEabrv,//users/localadmin/Dropbox/Bibfiles/refs,//users/localadmin/Dropbox/Bibfiles/strings}

\end{document}

%% file: figures/cycle_graph.tex
%!TEX root = ../tsp16chepuri_GPSD.tex

\begin{tikzpicture}
\def \radius {1.25cm}
\def \margin {25} 

\node[draw, circle, fill=gray] at ({360/4 * (0)}:\radius) {$x_3$};
\node[draw, circle, fill=gray] at ({360/4 * (1)}:\radius) {$x_2$};
\node[draw, circle, fill=gray] at ({360/4 * (2)}:\radius) {$x_1$};
\node[draw, circle, fill=gray] at ({360/4 * (3)}:\radius) {$x_N$};

\draw[ thick,-, >=latex] ({360/4 * (1 - 1)+\margin}:\radius) 
    arc ({360/4 * (1 - 1)+\margin}:{360/4 * (1)-\margin}:\radius);
\draw[ thick, -, >=latex] ({360/4 * (2 - 1)+\margin}:\radius) 
    arc ({360/4 * (2 - 1)+\margin}:{360/4 * (2)-\margin}:\radius);
 \draw[thick, >=latex] ({360/4 * (3 - 1)+\margin}:\radius) 
    arc ({360/4 * (3 - 1)+\margin}:{360/4 * (3)-\margin}:\radius);
\draw[ very thick, dotted, -, >=latex] ({360/4 * (4 - 1)+\margin}:\radius) 
    arc ({360/4 * (4 - 1)+\margin}:{360/4 * (4)-\margin}:\radius);

\end{tikzpicture}